\documentclass[reqno]{amsart} 
\usepackage[caption=false]{subfig}

\usepackage{graphicx}
\usepackage{enumitem}
\usepackage{array}
\usepackage{arydshln}
\usepackage[foot]{amsaddr}
\usepackage{ dsfont }
\usepackage{wrapfig}

\graphicspath{ {images/} }
\usepackage{ctable}
\usepackage{float}
\usepackage{comment} 
\usepackage{url}
\usepackage[titlenumbered,ruled]{algorithm2e}
\usepackage{upgreek} 

\usepackage{amssymb}
\usepackage{amstext}
\usepackage{amsthm}

\theoremstyle{plain}
\newtheorem{theorem}{Theorem}[section]
\newtheorem{lemma}[theorem]{Lemma}

\theoremstyle{definition}
\newtheorem{definition}[theorem]{Definition}

\theoremstyle{remark}

\newcommand*\dif{\mathop{}\!\mathrm{d}} 
\newcommand{\argmin}{\mathop{\mathrm{argmin}}} 
\newcommand{\prox}[1]{\text{prox}_{\mbox{\scriptsize\ensuremath{#1}}}} 
\newcommand{\proj}[1]{\text{proj}_{\mbox{\scriptsize\ensuremath{#1}}}} 
\newcommand{\sgn}[0]{~\normalfont \text{sgn}} 
\newcommand{\supp}[1]{~\normalfont \! \text{supp}(#1)} 

\usepackage{bbm}
\DeclareRobustCommand{\Nolla}{
  \text{\usefont{U}{bbold}{m}{n}0}} 
  
\DeclareMathOperator{\R}{\mathbb{R}}
\DeclareMathOperator{\N}{\mathbb{N}}
\DeclareMathOperator{\Z}{\mathbb{Z}}
\DeclareMathOperator{\diag}{\mathrm{diag}}

\renewcommand{\vec}[1]{\boldsymbol{#1}} 

\newcommand{\tp}{\vec{l}}
\renewcommand{\mp}{\vec{m}} 
\newcommand{\gen}{\psi} 
\newcommand{\cshear}{\psi_{a,s,\vec{l}}}
\newcommand{\Mas}{\vec{M}_{as}}
\newcommand{\SH}{\mathcal{SH}}
\DeclareMathOperator{\sh}{SH}
\DeclareMathOperator{\shD}{{\bf SH}}

\DeclareMathOperator{\Radon}{\mathcal{R}}
\DeclareMathOperator{\RadonD}{\vec{\mathcal{R}}}

\newcommand{\eg}{\emph{e.g.}}

\newcommand{\HL}[1]{\bf #1}

\begin{document}

\title[Sparse dynamic tomography]{Sparse dynamic tomography: A shearlet-based approach for iodine perfusion in plant stems}

\author[T.A.~Bubba]{Tatiana A.~Bubba}
\email{tatiana.bubba@helsinki.fi}

\author[T.~Heikkil\"{a}]{Tommi Heikkil\"{a}}
\email{tommi.heikkila@helsinki.fi}

\author[H.~Help]{Hanna Help}
\email{hanna.help@helsinki.fi}

\author[S.~Huotari]{Simo Huotari}
\email{simo.huotari@helsinki.fi}

\author[Y.~Salmon]{Yann Salmon}
\email{yann.salmon@helsinki.fi}

\author[S.~Siltanen]{Samuli Siltanen}
\email{samuli.siltanen@helsinki.fi}

\address[T.~A.~Bubba, T.~Heikkil\"{a}, S.~Siltanen]{Department of Mathematics and Statistics, University of Helsinki, 00014 Helsinki, Finland}
\address[T.~Heikkil\"{a}, H.~Help, S.~Huotari]{Department of Physics, University of Helsinki, 00014 Helsinki, Finland}
\address[H.~Help]{Finnish Food Authority, Finland}
\address[Y.~Salmon]{Institute for Atmospheric and Earth System Research, University of Helsinki, 00014 Helsinki, Finland}


\begin{abstract}
In this paper we propose a motion-aware variational approach to reconstruct moving objects from sparse dynamic data. 
The motivation of this work stems from X-ray imaging of plants perfused with a liquid contrast agent, aimed at increasing the contrast of the images and studying the phloem transport in plants over time.

The key idea of our approach is to deploy 3D shearlets as a space-temporal prior, treating time as the third dimension. The rationale behind this model is that a continuous evolution of a cartoon-like object is well suited for the use of 3D shearlets.
We provide a basic mathematical analysis of the variational model for the image reconstruction. The numerical minimization is carried out with primal-dual scheme coupled with an automated choice of regularization parameter. We test our model on different measurement setups: a simulated phantom especially designed to resemble a plant stem, with spreading points to simulate a spreading contrast agent; a measured agarose gel phantom to demonstrate iodide diffusion and geometry prior to imaging living sample; a measured living tree grown \textit{in vitro} and perfused  with a liquid  sugar-iodine-mix. 
The results, compared against a 2D static model, show that our approach provides reconstructions that capture well the time dynamic of the contrast agent onset and are encouraging to develop microCT as a tool to study phloem transport using iodine tracer. 
\end{abstract}

\maketitle

\section{Introduction}


Computed tomography (CT) revolutionized medicine, nondestructive testing and security application ever since the first working machines were introduced in the early 60s~\cite{Cormack63}. The gist of CT is that it allows a non-invasive recovery of the internal structure of an object by recording the attenuation of the X-rays which are propagated through the object from multiple direction views. The inverse problem of reconstructing the X-rays attenuation (and thus the object structure) from the recorded measures, or projections, is well understood when the object is static and a full collections of projections (which corresponds to a dense sampling of the projection views or angles) is available~\cite{natterer2001,nattererwubbeling01}.

Even though multiple alternative methods, such as magnetic resonance imaging (MRI) and positron emission tomography (PET), have been introduced since CT appeared, CT is still the most commonly used method in many application areas. It is especially from applications that still arise a deluge of mathematical problems related to concrete practical issues in tomography: for instance, limited data problems where comprehensive data can not be measured or dynamic tomography where the target is non-stationary over time. In particular, in many biomedical applications, the target is non-stationary and, in general, projection measurements are time-dependent: namely, the target changes in time between the recording of two consecutive projection images. Such problems are (severely) \textit{ill-posed}, in the sense that the solution might not exist or might not be unique, or it might be very sensitive to possible noise on the measurements and modelling errors. As such, the solution of ill-posed inverse problems needs always complementing the insufficient data with some prior information which might be available on the data~\cite{Engl1996}.

The focus in this paper is precisely on dynamic tomography: we investigate the reconstruction of moving objects from scarcely sampled data. Motion can be periodic, like in the case of the beating heart, or non-periodic, as the flow of a contrast agent in the blood vessels. Regular CT devices are already used dynamically, but the motion is usually not taken into account during the reconstruction task, especially when the changes are non-periodic. Indeed, ad hoc solutions exist for very slow or periodic movements: for instance, a heart can be satisfactorily imaged using gating~\cite{Ritman03}, and there exists techniques based on motion compensation~\cite{Bonnet03, Roux04, Li05, Katsevich10, Hahn14, Hahn16}. In both cases, the result is the reconstruction of a static CT image from dynamic data. Moreover, many existing techniques for dynamic tomography are based on filtered back-projection (FBP) algorithms, which requires a dense angular sampling. 

The motivation for this study comes from imaging of phloem transport in plant stems. Phloem transport has an impact on the whole ecosystem function, playing a crucial role in the whole plant physiology and its ability to adapt to changing environmental conditions~\cite{savage2016, salmon2019}. Indeed, the export from plants leaves of assimilated Carbon (C) as non-structural carbohydrates (of which sucrose is the main component for transport) happens via the phloem and accounts for approximately 80\% of the total assimilated C~\cite{lemoine2013}.

However, phloem is a highly sensitive and reactive tissue making it challenging to study. Furthermore, unlike water transport in trees and plant hydraulics for which well-established methods exist (\eg{}, sap flow, water potential measurement with pressure chamber), no routine methods has been developed for measuring phloem transport. 
Over the years, phloem transport has been traditionally measured using stable isotope pulse-labelling~\cite{epron2015} or estimated using point-dendrometers~\cite{mencuccini2013}. Recently, there has been an increasing interest in using imaging methods to avoid destructive sampling, and several methods have been successfully tested: MRI~\cite{windt2006}, \textit{in vivo} confocal laser scanning microscopy~\cite{cayla2015}, fluorescent tracer dye~\cite{savage2013} and PET~\cite{hubeau2015}.
However, none of these methods is currently widely used due to different limitations such as too low resolution, the use of radioactive material, or costly and rare instruments. 

Because it uses nonradioactive tracers and offer high resolution image, X-ray microtomography offers an interesting alternative to those methods. Furthermore, desktop X-ray microtomography is becoming more and more available in plant ecophysiology laboratories and recent years have shown a tremendous interest in the method to investigate \textit{in vivo} plant water and carbon status \cite{earles2018, knipfer2017}.

The main challenge to overcome when studying phloem transport with benchtop microCT is that full scans are too slow to capture the movement of the tracer (\eg{}, iodine) in the imaging windows. One way of reducing both the radiation dose and the duration of the imaging is to lower the number of scanning angles. This leads to the dynamic undersampled problem studied in this paper.

Sparse dynamic CT is lately receiving increasing attention. In~\cite{Niemi15} the authors proposed a variant of the level set method for a moving target imaged by one or several source-detector pairs: in their setting, the number of projection directions is severely limited because of spatial constraints. In~\cite{burger2017} motion is estimated using the optical flow methodology: the authors propose a joint variational approach that incorporates physical motion models and regularization of motion vectors for the case of severly undersampled data. The dimension reduction Kalman filter for undersampled dynamic data is proposed in~\cite{Hakkarainen19}, again in the scenario of multi-source arrangements. 
While these recent contributions achieve remarkable results, their setting does not comply with the our set up. Multiple sources are usually non available for microCT and the optical flow methodology is based on the assumption of brightness constancy, which fails in the case of iodine spreading in plant stems.

The method we propose is a variational approach based on shearlets~\cite{Kutyniok2012}, a recently proposed directional-aware representation systems for multidimensional data, to regularize the ill-posedness deriving from the loss of information in the data and the ongoing motion. In particular, rather than reconstructing each 2D time frame independently by regularizing only in the object space, our regularization terms uses 3D shearlets to take into account the space-time dynamics, namely the third dimension accounts for the evolution in time, yielding a (2+1)-dimensional model.

For static tomography, $\ell^1$-priors have been widely investigated and were used for the first time on X-ray measurements with real data~\cite{Siltanen03,Siltanen03part2}; $\ell^1$ shearlet-based regularization has been used to investigate numerous limited data problems (\eg{},~\cite{Bubba19} and the references therein) but the use of shearlets for dynamic tomography has never been investigated before. An initial proof of concept for this approach was carried out by the authors in~\cite{Bubba17}, but in this paper we adopt a different minimization algorithm complemented with the sparsity-based \textit{automatic} choice for the regularization parameter, introduced in~\cite{purisha2017}, which ease the well-known problem of the choice of the regularization parameter. Also, we complement the numerical experiments section with a brief analysis of the proposed approach in the continuous setting, showing the existence of minimizers for the proposed functional.

The main focus of the paper is, however, the computational side. We provide a detailed comparison of possible results in different measurement setups: a simulated phantom especially designed to resemble a plant stem, with five spreading points to simulate a spreading contrast agent; a measured agarose gel phantom to demonstrate iodide diffusion and geometry prior to imaging living sample; a measured living tree grown \textit{in vitro} and perfused  with a liquid  sugar-iodine-mix. The physical targets have been measured at the University of Helsinki microCT laboratory with a Phoenix Nanotom S scanner~\cite{suuronen2014}. 

The numerical experiments are restricted to a 2D setup, which allows to gain good insight into the problem, but the theoretical framework can be extended to 3D data, the only limitation being, at the moment, computational. The results in this paper, compared against a 2D static model, show that our approach provides reconstructions that capture well the time dynamic of the iodine onset, which is crucial to study phloem transport. These results are encouraging to develop microCT as a tool to study phloem transport using iodine tracer.

The rest of this paper is organized as follows. In section~\ref{sec:ContTh} we introduce a motion-aware model and formulate the reconstruction procedure in a time continuous setting. In Section~\ref{sec:DiscrTh} we discuss the discretization of our model as well as practical issues to solve the optimization problem. The results of the proposed method, tested on for a simulated and two physical phantom, are then presented in Section~\ref{sec:ExpAndRes}. Finally, we draw some conclusions and we give an outlook to future work in Section~\ref{sec:concl}.

\section{A shearlet-based variational model for dynamic tomography}
\label{sec:ContTh}
We start by recalling a well-defined time-dependent version of the Radon transform, first introduced in~\cite{burger2017}, which serves as a model for our dynamic X-ray tomography data. Then, we give a brief introduction to shearlet systems and discuss their properties when used in  regularization. We will need this to introduce our proposed motion-aware shearlet-based variational model for dynamic tomography, which combines the time-dependent Radon transform with $\ell_1$ shearlet-based regularization. We end the section by proving the existences of minimizers for such a variational model.

\subsection{Dynamic tomography.}
 X-ray tomography is based on the measurable attenuation of the radiation intensity as it travels through an object. The object is usually irradiated from different directions of view (or angles) and the resulting data, recorded by a detector and usually called sinogram, can be interpreted as a collection of line integrals of an unknown attenuation function $f(x)$. Mathematically, this is modeled through the notion of the 2D Radon transform~\cite{natterer2001}:
\begin{equation}
(\Radon f) (\theta, r) = \int_{\mathcal{L}(\theta, r)} f(x) \dif x,
\label{eq:RadonTr}
\end{equation}
where $\mathcal{L}(\theta, r) = \{ x \in \R^2 | \, x_1 \cos(\theta) + x_2 \sin(\theta)= r\}$ denotes a parametrization of the ray lines with $(\theta, r) \in [0,2\pi) \times \R$ and $f(x)$ is defined on $\R^2$. In practice, if the changes in intensities are measured from a sufficiently high number of directions and positions, the Radon transform can be inverted and the object's attenuation values $f(x)$ can be determined in a robust way~\cite{nattererwubbeling01}.

Contrary to~\eqref{eq:RadonTr}, where the object is assumed static, in dynamic X-ray tomography the attenuation function $f(x, t)$ depends also on the time $t \geq 0$. Hence, it is necessary to consider a time-dependent definition for the Radon transform. Following the approach introduced in~\cite{burger2017}, let $I(t)$ be the set of given measurements at time $t$ and $\Omega$ the bounded support of $f$ defined by 
\[
\Omega(t) = \{ x \in \Omega \; | \; \exists \, (\theta,r) \in I(t) \; : \; x_1 \cos(\theta) + x_2 \sin(\theta) = r\}. 
\]
Let us further assume that the time interval is fixed and bounded by the end point $T > 0$. Then, the time-dependent 2D Radon transform is defined to be the operator:
\begin{equation}
\Radon_{I(t)} : \; L^q \big( \Omega \times [0,T] \big) \longrightarrow 
    L^q \big(\cup_{t \in [0,T]} I(t) \times \{ t \} \big),
\qquad 
(\Radon_{I(t)}f(x, t))(\theta, r) = \int_{\mathcal{L}(\theta, r)} f(x, t) \dif x.
\label{eq:TimeRadonTr}
\end{equation}
In~\cite[prop.~2.2]{burger2017} the authors show that~\eqref{eq:TimeRadonTr} is a well-defined bounded linear operator for any $q \in [1,\infty)$. 

Similarly to the static case, the time-dependent attenuation $f(x,t)$ can be uniquely determined when the full sinogram is available for \textit{all} possible lines. However, due to the nature of our application, we are interested in the case of sparse angle measurements, namely the parameter $\theta$ is severely undersampled. In addition, measurements are typically corrupted by noise. Namely, our inverse problem is to reconstruct $f(x,t)$ from the noisy measurements $y_t^{\eta}$ given by
\begin{equation}
y_t^{\eta} \; := \; y_t + \eta \; = \; \Radon_{I(t)}f  
\label{eq:ourIP}    
\end{equation}
where the noise $\eta$ is modelled as a Gaussian process. In such a case, the uniqueness of the solution $f$ is not guaranteed anymore. Regularization strategies allow to obtain a stable solution for problems of the form~\eqref{eq:ourIP}. A well established approach is to look for minimizers of the regularized functional 
\begin{equation}
\frac{1}{2} \| \Radon_{I(t)}f - y_t^{\eta} \|_2^2 + \alpha \mathcal{P}(f)
\label{eq:ReguGenFun}    
\end{equation}
where $\mathcal{P}(f)$ is a penalty term, promoting desired properties in the solution, and the regularization parameter $\alpha > 0$ expresses the trade-off between the two parts. In particular, sparsity-promoting penalties of the form
\[
\mathcal{P}(f) = \| (\langle f, \psi_{\xi} \rangle)_{\xi} \|_1
\]
allow to suppress noise while preserving discontinuities like edges, a governing feature in images. Among many possible choices for the sparsifying system $(\psi_{\xi})_{\xi}$, we are interested in shearlets, a representation system proven to be optimal in resolving discontinuities. Before going into the details of our proposed penalty, we give a brief introduction to shearlet systems in the next subsection.

\subsection{Shearlets}
Shearlets are representation systems specifically designed to provide optimally sparse approximations of a specific class of signals, cartoon-like images. Here, we give a concise overview of their main properties and the reader is referred to the cited literature for more details and precise formulations of the described results.

\subsubsection{The continuous shearlet transform}
Continuous shearlets are obtained by applying three different operations to a well chosen generator function $\gen \in L^2(\R^n)$ obtaining affine systems of the form:
\[
\Big\{  \cshear = |\det \Mas|^{1/2} \gen \left( \Mas (\cdot - \tp) \right) \; : \; \Mas \in G, \, \tp \in \R^n  \Big\},
\]
where $G$ is a subset of the group $GL_n(\R)$ of invertible $n \times n$ matrices, with $n=2$ or $n=3$. For instance, when $n=2$, $\Mas$ is given by the composite matrix
\[
\Mas = \vec{A}_a^{-1}\vec{S}_s^{-1} = 
    \begin{pmatrix} a & 0 \\ 0 & \sqrt{a} \end{pmatrix}^{-1} \begin{pmatrix} 1 & s \\ 0 & 1 \end{pmatrix}^{-1}.
\]
Then, the continuous shearlet transform is defined as the mapping
\[
L^2(\R^2) \ni f \, \longmapsto  \,  \SH_\gen f (a,s,\tp) = \langle f, \cshear \rangle.
\]
The shearlet parameters $(a,s,\tp) \in \R_+ \times \R^{n-1} \times \R^n$ control anisotropic scaling, shearing and translation, respectively. Thus, $\SH_\psi$ analyzes the function $f$ around the location $\tp$ at different resolutions and orientations encoded by the scale and shearing parameters $a$ and $s$, respectively.

Throughout this work, it is assumed that $\gen$ has compact support, yielding a compactly supported shearlet system, as implemented in the ShearLab package \cite{Kutyniok2016} used in the numerical experiments in section~\ref{sec:results}. The precise construction of such systems is of rather technical nature and we refer the reader to \cite{Kutyniok2012} for the details. Here we remark only that 
it has been proven that, under mild conditions, a large class of compactly supported generators yields shearlet frames with controllable frame bounds \cite{Kittipoom2012, Kutyniok2012a}, implying for instance that a reconstruction of $f$ from its shearlet transform is possible.

One of the main results in shearlets theory states that the continuous shearlet transform allows resolving the wavefront set of distributions by analyzing the decay properties of the continuous shearlet transform. \cite{kutyniok2009, Grohs2011}. Roughly, the wavefront set of $f$ is resolved by distinguishing different decay rates of its continuous shearlet transform.

We omit here the precise statement that is of rather technical nature, also because it relies on the definition of \textit{cone-adapted} shearlet transform. The name refers to the fact that, to circumvent the unwanted directional bias that appears when $s \rightarrow \infty$, the frequency domain is cone-like partitioned. We introduce cone-adapted shearlets for the discrete case (see next section), since the implementation provided in ShearLab 3D relies on it.

\subsubsection{The discrete shearlet transform} \label{sec:discreteSHtransform}
Discrete shearlet systems can be formally obtained by sampling the parameter space $\R_+ \times \R^{n-1} \times \R^n$ on an appropriate discrete subset. Namely, $(a,s,\tp) \in \R_+ \times \R^{n-1} \times \R^n$ is replaced by $(j,k,\mp) \in \Z_+ \times \Z^{n-1} \times \Z^n$. Furthermore, cone-adapted discrete shearlet systems are derived by restricting the range for the shearing variable $k$, to allow a more equal distribution of the orientations. We introduce the formal definition by treating separately the $n=2$ and $n=3$ cases. 

\begin{definition}[2D Cone-Adapted Discrete Shearlets]
\label{def:ShearletSystem2d}
Let $\varphi, \psi \in L^2(\R^2)$ and $c = (c_1,c_2) \in R_+^2$. The \emph{cone-adapted discrete shearlet
system} is defined by
\begin{equation*}
\sh_{\text{2D}} (\varphi,\psi; c) = \Phi(\varphi; c_1) \cup \Psi(\psi; c) \cup \widetilde{\Psi}(\widetilde{\psi}; c),
\end{equation*}
where
\begin{alignat*}{2}
\Phi(\varphi; c_1) & :=  \Big\{ \varphi_{\mp}  && := \varphi(\cdot - c_1\mp) :\mp \in \Z^2  \Big\},\\
\Psi(\psi; c) & := \Big\{ \psi_{j,k,\mp}  && := 2^{(3j)/4}\psi(\vec{S}_k \vec{A}_{2^j}\cdot - \vec{M}_c\mp ): j\in \N_0, k \in \Z, |k| \leq   2^{\left\lceil j/2 \right\rceil}, \mp \in \Z^2 \Big\},\\
\widetilde{\Psi}(\widetilde{\psi}; c) & :=\Big\{ \widetilde{\psi}_{j,k,\mp}  && := 2^{(3j)/4}\widetilde{\psi}(\vec{S}_k^T \vec{\widetilde{A}}_{2^j}\cdot - \widetilde{\vec{M}}_c \mp ): j\in \N_0, k \in \Z, |k| \leq
2^{\left\lceil j/2 \right\rceil}, \mp \in \Z^2 \Big\},
\end{alignat*}
with $\widetilde{\psi}(x_1,x_2) := \psi(x_2,x_1)$, $\vec{A}_{2^j} = \diag(2^{j},2^{j/2}) \in \R^{2 \times 2}$, $\vec{\widetilde{A}}_{2^j} = \diag(2^{j/2},2^{j}) \in \R^{2 \times 2}$, $\vec{M}_c = \diag(c_1, c_2)\in \R^{2 \times 2}$ and $\widetilde{\vec{M}}_c = \diag(c_2, c_1)\in \R^{2 \times 2}$. For ease of notation we introduce the index set 
\begin{align*}
\Lambda := \N_0 \times \{-2^{\left\lceil j/2 \right\rceil}, \ldots, 2^{\left\lceil j/2 \right\rceil} \} \times \Z^2.
\end{align*}
The \emph{cone-adapted discrete shearlet transform} is then defined as the mapping
\begin{equation*}
    L^2(\R^2) \ni  f \longmapsto \sh_{\gen,\varphi} f (m', (j,k,\mp),(\widetilde{j},\widetilde{k},\widetilde{\mp}))  = 
    \left( \langle f, \varphi_{m'} \rangle, \, \langle f,\gen_{j,k,\mp} \rangle, \, \langle f,\widetilde{\gen}_{\, \widetilde{j},\widetilde{k},\widetilde{\mp}} \rangle \right).
\end{equation*}
with $(m', (j,k,\mp),(\widetilde{j},\widetilde{k},\widetilde{\mp})) \in \Z^2 \times \Lambda \times \Lambda$.
\end{definition}
In the previous definition, $\varphi$ is referred to as the \emph{shearlet scaling function} and it is associated to the low frequency part, while the function $\gen$ is referred to as \emph{shearlet generator}. The corresponding systems $\Psi(\psi) $ and $\widetilde{\Psi}(\widetilde{\psi})$ essentially differ in the reversed roles of the input variables and therefore correspond to the horizontal and vertical conic region, respectively.

Definition \ref{def:ShearletSystem2d} can be extended quite naturally to the 3D case: the frequency domain is partitioned into three pairs of pyramids (described by $\Psi(\psi; c), \, \widetilde{\Psi}(\widetilde{\psi}; c)$ and $\breve{\Psi}(\breve{\psi})$) and a centered cube (described by $\Phi(\varphi; c_1)$). Once more, the partitioning of the frequency domain into pyramids allows to restrict the range of the shear parameters, which is the key to provide an almost uniform treatment of different directions in a sense of a good approximation to rotation.

\begin{definition}[3D Pyramid-Adapted Discrete Shearlets]
\label{def:ShearletSystem3d}
Let $\varphi, \psi \in L^2(\R^3)$ and $c = (c_1,c_2) \in \R_+^2$. The \emph{pyramid-adapted discrete shearlet
system} is defined by
\begin{equation*}
\sh_{\text{3D}} (\varphi,\psi; c) = \Phi(\varphi; c_1) \cup \Psi(\psi; c) \cup \widetilde{\Psi}(\widetilde{\psi}; c) \cup \breve{\Psi}(\breve{\psi}; c),
\end{equation*}
where
\begin{alignat*}{2}
\Phi(\varphi; c_1) & :=  \Big\{ \varphi_{\mp}  && := \varphi(\cdot - c_1\mp) :\mp \in \Z^3  \Big\},\\
\Psi(\psi; c) & := \Big\{ \psi_{j,k,\mp}  && := 2^j\psi(\vec{S}_k \vec{A}_{2^j}\cdot - \vec{M}_c\mp ): j\in \N_0, k \in \Z^2, |k| \leq   2^{\left\lceil j/2 \right\rceil}, \mp \in \Z^3 \Big\},\\
\widetilde{\Psi}(\widetilde{\psi}; c) & :=\Big\{ \widetilde{\psi}_{j,k,\mp}  && := 2^j\widetilde{\psi}(\vec{\widetilde{S}}_k \vec{\widetilde{A}}_{2^j}\cdot - \widetilde{\vec{M}}_c \mp ): j\in \N_0, k \in \Z^2, |k| \leq
2^{\left\lceil j/2 \right\rceil}, \mp \in \Z^3 \Big\}, \\
\breve{\Psi}(\breve{\psi}; c) & := \Big\{ \breve{\psi}_{j,k,\mp}  && := 2^j\breve{\psi}(\vec{\breve{S}}_k \vec{\breve{A}}_{2^j}\cdot - \vec{\breve{M}}_c\mp ): j\in \N_0, k \in \Z^2, |k| \leq   2^{\left\lceil j/2 \right\rceil}, \mp \in \Z^3 \Big\},
\end{alignat*}
with $\vec{A}_{2^j} = \diag(2^{j},2^{j/2}, 2^{j/2}) \in \R^{3 \times 3}$, $\vec{\widetilde{A}}_{2^j} = \diag(2^{j/2},2^{j}, 2^{j/2}) \in \R^{3 \times 3}$, $\vec{\breve{A}}_{2^j} = \diag(2^{j/2}, 2^{j/2},2^{j}) \in \R^{3 \times 3}$, $\vec{M}_c = \diag(c_1, c_2, c_2)\in \R^{3 \times 3}$,  $\widetilde{\vec{M}}_c = \diag(c_2, c_1, c_2)\in \R^{3 \times 3}$ and $\breve{\vec{M}}_c = \diag(c_2, c_2, c_1)\in \R^{3 \times 3}$. The shearing matrices are given by:
\[
\vec{S}_k = \begin{pmatrix} 1 & k_1 & k_2 \\ 0 & 1 & 0 \\ 0 & 0  & 1 \end{pmatrix},
\qquad
\vec{\widetilde{S}}_k = \begin{pmatrix} 1 & 0 & 0 \\ k_1 & 1 & k_2  \\ 0 & 0  & 1 \end{pmatrix}
\qquad \text{and} \qquad 
\vec{\breve{S}}_k = \begin{pmatrix} 1 & 0 & 0 \\ 0 & 1 & 0 \\ k_1 & k_2  & 1 \end{pmatrix}.
\]
The \emph{pyramid-adapted discrete shearlet transform} is then defined as the mapping
\begin{equation*}
    L^2(\R^3) \ni  f \longmapsto \sh_{\gen,\varphi} f (m', (j,k,\mp),(\widetilde{j},\widetilde{k},\widetilde{\mp}), (\breve{j},\breve{k},\breve{\mp}))  = 
    \left( \langle f, \varphi_{m'} \rangle, \, \langle f,\gen_{j,k,\mp} \rangle, \, \langle f,\widetilde{\gen}_{\, \widetilde{j},\widetilde{k},\widetilde{\mp}} \rangle, \langle f,\breve{\gen}_{\, \breve{j},\breve{k},\breve{\mp}} \rangle \right).
\end{equation*}
with $(m', (j,k,\mp),(\widetilde{j},\widetilde{k},\widetilde{\mp}), (\breve{j},\breve{k},\breve{\mp})) \in \Z^3 \times \Lambda \times \Lambda \times \Lambda$, where $\Lambda := \N_0 \times \{-2^{\left\lceil j/2 \right\rceil}, \ldots, 2^{\left\lceil j/2 \right\rceil} \}^2 \times \Z^3$.
\end{definition}

There are various extensions and refinements of these basic definitions available in the literature: we refer the interested reader to \cite{Kutyniok2012} and references therein. 

From Definition~\ref{def:ShearletSystem3d}, it is straightforward to show that the shearlet transform is a bounded linear operator. This result will be needed later on to prove the existence of minimizers for the proposed model.

\begin{lemma} \label{lemma:shearletBound}
Let $f \in L^q \left(\R^3\right)$ for $q > 4$ be a compactly supported function and $\varphi, \psi \in L^2\left(\R^3\right)$ compactly supported functions defining a pyramid-adapted discrete shearlet system in sense of Definition~\ref{def:ShearletSystem3d}. Then the $\ell^1$-norm of the shearlet transform
$\sh_{\gen,\varphi}: L^2 \left( \R^3 \right) \longrightarrow \ell^2(\Z^3 \times \Lambda \times \Lambda \times \Lambda)$ is bounded, namely:
\[
\|\sh_{\gen,\varphi} f \|_1 \; < \; C \|f\|_q (\|\varphi\|_p + \|\psi\|_p + \|\widetilde{\psi}\|_p + \|\breve{\psi}\|_p),
\]
for some $0 < C < \infty$ and $p < 4/3$ the H\"{o}lder conjugate of $q$.
\end{lemma}

\begin{proof}
Let's start by considering $\psi$. Since both $f \in L^q \left(\R^3\right)$ and $\psi \in L^2(\R^3)$ are compactly supported, there exists at most $U<\infty$ translations $T_{\mp}(x) := x - \mp$ such that $\supp{f} \cap \supp{\psi \circ T_{\mp}} \neq \emptyset$. Then, the support of a scaled shearlet $\psi_{j,k,\cdot}$ requires at most $U$ translations as well.
The shearing parameter is limited by the scale parameter $j$ and, therefore, those contribute for at most $2\cdot 2 \cdot 2^{\left\lceil j/2 \right\rceil} + 2$ terms.

Now, since $\psi \in L^2(\R^3)$ is compactly supported, we have that $\psi \in L^p\left(\R^3\right)$ for $p \leqslant 2$. Using the fact that translations do not alter $L^p$-norms, it is easy to check that:
\begin{equation*}
\|\psi_{j,k, \cdot}\|_p \; = \;
    \left(\int_{\R^3} |2^j \psi\left(\vec{S}_k \vec{A}_{2^j} x \right)|^p \dif x \right)^\frac{1}{p} \;
\; = \; 2^{j - \frac{2j}{p}} \|\psi\|_p.
\end{equation*}
by using the change of variables $y = \vec{S}_k \vec{A}_{2^j} x$ and  $|\det(\vec{A}_{2^j}^{-1} \vec{S}_k^{-1})| = |\det( \vec{S}_k \vec{A}_{2^j})|^{-1} = 2^{-2j}$. Then, by estimating the inner products with H\"{o}lder's inequality and combining the previous results we get:
\begin{align*}
\sum_{j=0}^\infty \; \sum_{|k| \leq 2^{\left\lceil j/2 \right\rceil}} \; \sum_{\mp \in \Z^3} |\langle f, \psi_{j,k,\mp} \rangle | 
&\; \leq \; \sum_{j=0}^\infty \; \sum_{|k| \leq 2^{\left\lceil j/2 \right\rceil}} \; \sum_{\mp \in \Z^3} \| f \|_q \|\psi_{j,k,\mp} \|_p \\ 
&\; \leq \; \sum_{j=0}^\infty \; \left( 2^{\lceil j/2 \rceil + 2} + 2 \right) U \,  2^{j - \frac{2j}{p}} \|f\|_q \|\psi\|_p \\
&\; < \; 8 U \|f\|_q \|\psi\|_p \sum_{j=0}^\infty \; 2^{\frac{3j}{2} - \frac{2j}{p}}  \\[0.1em]
&\; = \; C_1 \, \|f\|_q \, \|\psi\|_p,
\end{align*}
since $p < 4/3 \Rightarrow 3j/2 - 2j/p < 0$ and the geometric series converges.

Given how the different directions are defined, the same calculations can be replicated for $\widetilde{\psi}$ and  $\breve{\psi}$, respectively, in place of $\psi$. Finally, similar calculations can be done for the scaling function $\varphi$. Indeed, $\varphi \in L^p(\R^3)$ since it is compactly supported and, similarly to $\psi$, it overlaps the support of $f$ for at most $U'$ translations. Then, it easily follows that:
\[
\sum_{\mp \in \Z^3} |\langle f, \varphi(\cdot - \mp ) \rangle | \leqslant U' \|f\|_q \|\varphi\|_p.
\]
By simply combining the previous estimates, we get:
\[
\|\sh_{\gen,\varphi} f \|_1 \; < \; C \|f\|_q (\|\varphi\|_p + \|\psi\|_p + \|\widetilde{\psi}\|_p + \|\breve{\psi}\|_p),
\]
for some $0 < C < \infty$ depending on the support of $f$.
\end{proof}

Observe that, in particular, the previous lemma implies that $\|\sh_{\gen,\varphi} \cdot \|_1 : L^q(\Omega) \longrightarrow \R$ is a continuous sublinear functional for a compact set $\Omega \subset \R^3$.

We conclude this excursion into discrete shearlet theory by pointing out another fundamental result: shearlets are an optimal sparsifying transform for the class $\mathcal{E}^2(\R^n)$ of cartoon-like images. Namely, it can be proven that 
\[
\| f - f_D \|^2 \lesssim D^{-\frac{2}{n-1}}
\qquad \text{ for all } \quad f \in \mathcal{E}^2(\R^n),
\]
where $f_D$ denotes the $D$-term approximation using the $D$ largest shearlets coefficients. The precise statements can be found in \cite{kutyniokch5}. The previous inequality reveals that shearlets allow to attain what is proven to be the optimal rate for cartoon-like functions \cite{Donoho2001}. The same rate is also achieved by other anisotropic systems such as curvelets \cite{Candes2002}, but out of reach for isotropic wavelet systems \cite{Mallat2009}. An intuitive explanation for this is that the anisotropic scaling and the shearing allow to capture geometric features more efficiently than isotropic wavelet systems.

\subsection{A Shearlet-based motion-aware model}
In our particular problem, we are interested in tracking down iodine spreading into a plant's stem. This means, in particular, that the changes over time only affect limited sections of the target: the static background and, even more importantly, edges are not affected by motion and can be preserved by choosing an appropriate regularization strategy. Preserving these features is also important for differentiating between motion and measurement noise.

As we refreshed in the previous subsection, anisotropic structures like edges can be efficiently represented with a multivariate basis or frame, such as shearlets. Regularizing each time step independently by applying the same 2D shearlet transform for each time step would mostly likely lead to a denoised reconstruction which completely disregards the time evolution.   
However, the directional properties of shearlets are ideally suited to well represent a continuous evolution of a cartoon-like object. This suggest that considering 3D shearlet systems, where the third dimension is time evolution, would not only regularize the ill-posedness given by the scarce data and noise, but would also add an additional constraint that ``link'' the subsequent time frames.  This amounts to choosing the penalty term $\mathcal{P}(f)$ in~\eqref{eq:ReguGenFun} equal to:
\[
\mathcal{P}(f) = \| \sh_{\gen,\varphi} f\|_1,
\]
where $\sh_{\gen,\varphi}$ represents the 3D shearlet transform as in Definition~\ref{def:ShearletSystem3d}, leading to the following motion-aware variational model:
\begin{equation} 
\argmin_{f \geq 0} \quad \int_0^T \frac{1}{2} \big\|\Radon_{I(t)} f(\cdot,t) - y_t^{\eta} \big\|_2^2 \, \dif t \; + \; \alpha \|\sh_{\gen, \varphi} f\|_1.
\label{eq:functionalJ}
\end{equation}
From the perspective of image reconstruction the 3D shearlet penalty term acts as a temporal constraint allowing, at the same time, to isolate and suppress noise from motion and to preserve edges across time. The constraint $f \geq 0$ imposes the non-negativity for the solution, a prior given by the physics of the problem, given that $f$ is understood as a measure of X-ray attenuation. 

We conclude the section by showing the existance of a minimizer for~\eqref{eq:functionalJ}.

\begin{theorem}[Existance of a minimizer]
\label{th:minimizer}
Let $\Radon_{I(t)}: L^2\left(\Omega \times [0,T]\right) \longrightarrow L^2 \left(\cup_{t \in [0,T]} I(t) \times \lbrace t \rbrace \right)$ be the time dependent Radon transform as defined in~\eqref{eq:TimeRadonTr} and $\sh_{\gen, \varphi}: L^2\left(\Omega \times [0,T]\right) \longrightarrow \ell^2(\Z^3 \times \Lambda \times \Lambda \times \Lambda)$ a compactly supported discrete pyramid-adapted shearlet transform and $\alpha > 0$. Then the functional

\begin{equation*} 
J(f) \; = \; \int_0^T \big\|\Radon_{I(t)} f(\cdot,t) - y_t^{\eta}\big\|_2^2 \dif t + \alpha \|\sh_{\gen, \varphi} f\|_1
\end{equation*}
accepts a minimizer in the constraint set  $\mathcal{C}_F := \{ f \in L^2\left( \Omega \times [0,T] \right) : \ 0 \leq f(x,t) \leq F \ a.e. \}$ for any $F > 0$.
\end{theorem}

\begin{proof}
The constraint set $\mathcal{C}_F$ is clearly convex and also closed in $L^2$ (closedness follows, for instance, from the Riesz-Fisher theorem~\cite[p. 148]{royden2010}) and since $L^\infty (\Omega \times [0,T])$ is closed in $L^2 (\Omega \times [0,T])$. 
Moreover $F \geq \|f\|_\infty \geq C_q \|f\|_q$, for any $q \geq 1$ and constant $C_q$.

To prove the continuity of $J$, let's consider its two terms separately. The first part of $J$ is continuous in $L^2$ thanks to the continuity of the 2-norm and~\cite[prop. 2.2]{burger2017}. In particular this implies the continuity in $\mathcal{C}_F$ also with respect to $\|\cdot\|_q$. Given~lemma \ref{lemma:shearletBound}, the second part of $J$ is also continuous in $\mathcal{C}_F$ with respect to $\|\cdot\|_q$.

It is easy to check that the functional $J$ is proper and convex. Therefore it must have a minimizer in the bounded set $\mathcal{C}_F$ \cite[p. 177]{royden2010}.
\end{proof}

We remark that, in practical applications, the upper bound $F$ on the attenuation values is anyway enforced by the digital setting of the recording detector. In general, detectors produce values in the interval $[0, 4095]$ in a 12 bit A/D conversion.

\section{Numerical Reconstruction Framework}
\label{sec:DiscrTh}
In this section we discuss the discretization of the proposed model~\eqref{eq:functionalJ} and how to practically carry out the minimization. For this, we will be employing a primal-dual scheme endowed with an automatic choice for the regularization parameter.  

\subsection{Discretization of the model}
It is well known that in the discrete setting the solution of the Fredholm integral equation like~\eqref{eq:RadonTr} leads to the following system of linear equations:
\begin{equation}
\RadonD \vec{f} = \vec{y^{\eta}},
\label{eq:StaticRadonDiscr}    
\end{equation}
where $\vec{f} \in \R^N$ is a discrete representation of the unknown object, using pixels as a basis, and stacked in a vector form. Similarly, $\vec{y^{\eta}} \in \R^M$ is the vector of the measured noisy data, whose dimension depends on the number of projection views and sensors in the detector. The matrix $\RadonD \in \R^{M \times N}$ is the discrete counterpart of the (static) Radon transform, containing information about the specific geometry set up, which sections of the object are penetrated and how much. 

Similarly to the static case, we aim at representing also the dynamic case~\eqref{eq:TimeRadonTr} as a system of linear equations. The measurements $y_t^{\eta}$ are taken at certain time instances $\tau = 1, 2, \ldots, T$ during a fixed time period $[0,T]$. At each time instance $\tau$ we have a system of linear equations like~\eqref{eq:StaticRadonDiscr}, namely:
\begin{equation}
\RadonD_{\tau} \vec{f}_{\tau} = \vec{y}_{\tau}^{\vec{\eta}}
\qquad \text{for each } \;
\tau = 1, 2, \ldots, T.
\label{eq:StaticRadonDiscrTau}    
\end{equation}
Now, it is easy to obtain the discrete dynamic system of linear equations by stacking, for each $\tau \in \{1, \ldots, T\}$, the vectors $\vec{f}_{\tau}$ and $\vec{y}_{\tau}^{\vec{\eta}}$ into longer column vectors $\widetilde{\vec{f}} \in \R^{NT}$ and $\widetilde{\vec{y^{\eta}}} \in \R^{MT}$, respectively; the matrix for the projection of all time instances is given by a block diagonal matrix $\widetilde{\RadonD} \in \R^{NT \times MT}$. Hence, we get the following system of linear equations:  
\begin{equation} 
\widetilde{\RadonD} \vec{\widetilde{f}} \; = \;
\begin{bmatrix}
\RadonD_1 & \Nolla & \cdots & \Nolla \\[0.2em]
\Nolla & \RadonD_2 & \cdots & \Nolla \\[0.2em]
\vdots & \vdots & \ddots & \vdots \\[0.2em]
\Nolla & \Nolla & \cdots & \RadonD_T
\end{bmatrix}
\begin{bmatrix}
\vec{f}_1 \\[0.2em]
\vec{f}_2 \\[0.2em]
\vdots \\[0.2em]
\vec{f}_T 
\end{bmatrix}
\; = \;
\begin{bmatrix}
\vec{y}_1^{\vec{\eta}} \\[0.2em]
\vec{y}_2^{\vec{\eta}} \\[0.2em]
\vdots \\[0.2em]
\vec{y}_T ^{\vec{\eta}}
\end{bmatrix}
\; = \; \widetilde{\vec{y^{\eta}}}
\label{eq:Rf=m 3D}
\end{equation}
where $\Nolla \in \R^{M \times N}$ is the zero matrix. Finally, the discretization into time steps for the entire variational model~\eqref{eq:functionalJ} yields:
\begin{equation}
\argmin_{\vec{\widetilde{f}} \geq 0} \quad 
    \bigg\{
    \frac{1}{2} \big\|\widetilde{\RadonD} \vec{\widetilde{f}} -  \widetilde{\vec{y^{\eta}}} \big\|_2^2 \; + \;
    \alpha \, \big\| \shD_{\text{3D}} \vec{\widetilde{f}} \big\|_1
    \bigg\}
\label{eq:MotionAwareVMdiscr}    
\end{equation}
or, equivalently,
\begin{equation}
\argmin_{\vec{\widetilde{f}} \geq 0} \quad 
    \Bigg\{
    \sum_{\tau=1}^T
    \frac{1}{2} \big\|\RadonD_{\tau} \vec{f}_{\tau} -  \vec{y}_{\tau}^{\vec{\eta}} \big\|_2^2 \; + \;
    \alpha \, \big\| \shD_{\text{3D}} \vec{\widetilde{f}} \big\|_1
    \Bigg\}
\label{eq:MotionAwareVMdiscr2}    
\end{equation}
where the the inequalities are meant component-wise. Without the temporal connection given by the 3D shearlet regularization term,
each time frame can also be considered just as a static  tomographic problem and each time instance can be solved \textit{independently} using lower dimensional shearlets or wavelets, leading to $T$ disjoint subproblems: 
\begin{equation}
\argmin_{\vec{f}_1, \ldots, \vec{f}_T \geq 0} \quad 
    \bigg\{
    \frac{1}{2} \big\|\RadonD_{\tau} \vec{f}_{\tau} -  \vec{y}_{\tau}^{\vec{\eta}} \big\|_2^2 \; + \;
    \alpha \, \big\| \mathbf{\Psi} \vec{f}_{\tau} \big\|_1
    \bigg\}
\qquad
\text{for each }
\quad \tau=1,\ldots, T
\label{eq:StaticVMdiscr}    
\end{equation}
where $\mathbf{\Psi}$ represent either a 2D shearlet transform $\shD_{\text{2D}}$ or a 2D wavelet transform $\mathbf{W}_{\text{2D}}$. This static variants will be used in section~\ref{sec:ExpAndRes} for comparison with the proposed motion-aware shearlet-based model.

\subsection{Minimization algorithm}
While theorem~\ref{th:minimizer} ensures the existence of a minimizer, the practical reconstruction process from discrete measurements requires a specific algorithm to address, in particular, the non-differentiability of the $\ell_1$-norm.

To this end, we use a primal-dual approach, called \emph{primal-dual fixed point} (PDFP)~\cite{chen2016}, recently introduced to generalize the well-known iterative soft-thresholding algorithm (ISTA)~\cite{daubechies2004} to the case of frames and allowing constraints on convex sets. In addition, the PDFP algorithm is endowed with an automatic strategy for the choice of the regularization parameter, called \textit{controlled wavelet domain sparsity} (CWDS)~\cite{purisha2017}, which we extended also to the case of shearlets. The non-negativity constraint, which is known to have a strong positive effect on the reconstruction quality, is easily included in PDFP and the automated sparsity tuning provides robustness to the choice of regularization parameter. It also reduces any accidental bias between the different shearlet and wavelet systems since those require uniquely weighted regularization terms.

In details, the PDFP algorithm is well suited for solving minimization problems of the form:
\[
\argmin_{\vec{x} \in \R^\nu} J_1(\vec{x}) + J_2\big(\mathbf{B} \vec{x}\big) + J_3(\vec{x}),
\]
where $J_1,J_2,J_3$ are proper lower semi-continuous convex functions, $J_1$ is differentiable with $1/L$-Lipschitz continuous gradient and $\mathbf{B}:\R^{\nu} \longrightarrow \R^{\mu}$ is a bounded linear map. This applies exactly to the case of equations~\eqref{eq:MotionAwareVMdiscr2} and~\eqref{eq:StaticVMdiscr} by choosing $J_1(\cdot) = \frac{1}{2}\|\widetilde{\RadonD} \cdot - \widetilde{\vec{y^{\eta}}}\|_2^2$ or $J_1(\cdot) = \frac{1}{2}\|\RadonD_{\tau} \cdot - \vec{y}_{\tau}^{\vec{\eta}}\|_2^2$ for each $\tau =1,\ldots,T$, $J_2(\cdot) = \alpha\|\cdot\|_1$ with $\mathbf{B}$ equal to $\shD_{\text{3D}}$, $\shD_{\text{2D}}$ or $\mathbf{W}_{\text{2D}}$, and $J_3$ as the non-negativity constraint given by the indicator function of the non-negative orthant:
\begin{equation*}
\iota_{\R_+}(x) = \left\lbrace
\begin{array}{ll}
0       &\qquad \text{if } \ x \geq 0 \\
+\infty &\qquad \text{if } \ x < 0.
\end{array} \right.
\end{equation*}
The $(i + 1)$th iterate, with $i = 0, 1, 2, \ldots$, of the PDFP algorithm is computed via the following iteration scheme:
\begin{equation} 
\text{(PDFP)} \ \left\lbrace
\begin{array}{cl}
\vec{d}^{(i+1)} &= \; \prox{\gamma J_3} \big(\vec{x}^{(i)} - \gamma \nabla J_1(\vec{x}^{(i)}) - \lambda \mathbf{B}^T \vec{v}^{(i)} \big), \\[0.35em]
\vec{v}^{(i+1)} &= \; \big(\mathbbm{1} - \prox{\frac{\gamma}{\lambda} J_2}\big) \big(\mathbf{B} \vec{d}^{(i+1)} + \vec{v}^{(i)} \big), \\[0.35em]
\vec{x}^{(i+1)} &= \; \prox{\gamma J_3} \big(\vec{x}^{(i)} - \gamma \nabla J_1(\vec{x}^{(i)}) - \lambda \mathbf{B}^T \vec{v}^{(i+1)} \big),
\end{array} \right.
\label{eq:genPDFP}
\end{equation}
where $0 < \lambda < 1/ \lambda_{\max} (\mathbf{B B}^T)$ and $0 < \gamma < 2L$, $\mathbbm{1}$ is the identity operator and $\prox{}$ denotes the so-called proximity operator. 

Here, $\lambda_{\max} (\mathbf{B B}^T)$ is the largest eigenvalue of the square matrix $\mathbf{B B}^T$, $1/L$ is the Lipschitz constant for $\nabla J_1$.
In particular, as advised in the original paper~\cite{chen2016}, $\gamma$ and $\lambda$ should be chosen close to their upper limits $2L$ and $1/\lambda_{\max}$. Due to their linearity, the operators $\widetilde{\RadonD}$ and $\RadonD_{\tau}$, for each $\tau=0,1,\ldots, T$ can be normalized, yielding $1/L = 1$. Regarding $\lambda_{\max} (\mathbf{B B}^T)$, for the comparison case of 2D wavelets $\mathbf{W}_{\text{2D}}$, we will consider an orthogonal family, namely $\lambda_{\max} (\mathbf{W}_{\text{2D}}\mathbf{W}_{\text{2D}}^T) =1$. Regarding shearlets, both $\shD_{\text{3D}}$ and $\shD_{\text{2D}}$ represents compactly supported shearlet systems, which in general are frames. We have the following lemma.

\begin{lemma} 
For a linear transformation $\mathbf{B}: \R^{\nu} \longrightarrow \R^{\mu}$ with upper frame bound $u$, the largest eigenvalue of $\mathbf{B B}^T$ is bounded:
\[
\lambda_{\max} (\mathbf{B B}^T) \; = \; 
    \lambda_{\max} (\mathbf{B}^T \mathbf{B}) \; \leq \; u^2.
\]    
\label{lem:eigenbound}
\end{lemma}

\begin{proof}
The equality follows from the fact that $\mathbf{BB}^T$ and $\mathbf{B}^T \mathbf{B}$ must have the same non-zero eigenvalues. Both are also positive and semi-definite, hence the eigenvalues must be non-negative and at least one eigenvalue must be non-zero unless $\mathbf{B}$ is a zero matrix.

Now, let $\bar{\lambda} > 0$ be the largest eigenvalue of $\mathbf{B}^T \mathbf{B}$ and $\vec{g} \in \R^n$ the corresponding eigenvector. Then:
\[
    u^2 \|\vec{g}\|^2 \; \geq \; \|\mathbf{B} \vec{g}\|^2
    \; = \; \vec{g}^T \mathbf{B}^T \mathbf{B} \vec{g} 
    \; = \; \vec{g}^T \bar{\lambda} \vec{g} 
    \; = \; \bar{\lambda} \|\vec{g}\|^2
\]
which ends the proof.
\end{proof}
In particular, the upper frame bound for the 2D and 3D compactly supported shearlets employed in section~\ref{sec:ExpAndRes} is 1. Hence, by lemma~\ref{lem:eigenbound} we have 
$\lambda_{\max} (\shD \, \shD^T) \leq 1$. Therefore, the parameters are limited by $\gamma < 2$ and $\lambda < 1$.

Regarding the proximity operator, for any constant $\gamma > 0$, the corresponding proximity operator of the non-negativity constraint $\gamma \iota_{\R_+}(x)$ is simply a projection onto the non-negative orthant. The proximity operator of the term $\alpha \|\cdot\|_1$ is the so called soft-thresholding operator $S_\alpha$~\cite{daubechies2004} given by:
\[
S_\alpha(x) := \left\lbrace
\begin{array}{ll}
x - \alpha &\qquad \text{ if } \ \ x > \alpha, \\
0          &\qquad \text{ if } \ |x| \leqslant \alpha, \\
x + \alpha &\qquad \text{ if } \ \ x < -\alpha.
\end{array} \right.
\]
Both proximity operators are meant to be applied component-wise to vectors.
Finally, we have 
\[
\nabla \bigg( \frac{1}{2} \Big\|\widetilde{\RadonD}\widetilde{\vec{f}} - \widetilde{\vec{y^{\eta}}} \Big\|_2^2 \bigg) 
\; = \; 
\widetilde{\RadonD}^T \widetilde{\RadonD} \widetilde{\vec{f}} - \widetilde{\RadonD}^T \widetilde{\vec{y^{\eta}}} 
\]
and similarly, for the static case, $\nabla \big( \frac{1}{2}\|\RadonD_{\tau} \vec{f}_{\tau} - \vec{y}_{\tau}^{\vec{\eta}}\|_2^2 \big) 
\; = \; 
\RadonD_{\tau}^T \RadonD_{\tau} \vec{f}_{\tau} - \RadonD_{\tau}^T \vec{y}_{\tau}^{\vec{\eta}}$.
In conclusion, the application of the PDFP method to equations~\eqref{eq:MotionAwareVMdiscr2} yields the following iteration scheme:
\begin{equation}  
\left\lbrace
\begin{array}{cl}
\widetilde{\vec{d}}^{(i+1)}& = \proj{\R_+^{NT}} \big(\widetilde{\vec{i}}^{(i)} - \gamma (\widetilde{\RadonD}^T \widetilde{\RadonD}\widetilde{\vec{f}}^{(i)} - \widetilde{\RadonD}^T\widetilde{\vec{y^{\eta}}}) - \lambda \shD_{\text{3D}}^T \widetilde{\vec{v}}^{(i)} \big), \\[0.5em]
\widetilde{\vec{v}}^{(i+1)}& = \big(\mathbbm{1} - S_{\alpha \frac{\gamma}{\lambda}}\big) \big(\shD_{\text{3D}} \widetilde{\vec{d}}^{(i+1)} + \widetilde{\vec{v}}^{(i)} \big), \\[0.5em]
\widetilde{\vec{f}}^{(i+1)}& = \proj{\R_+^{NT}} \big(\widetilde{\vec{f}}^{(i)} - \gamma (\widetilde{\RadonD}^T \widetilde{\RadonD}\widetilde{\vec{f}}^{(i)} - \widetilde{\RadonD}^T\widetilde{\vec{y^{\eta}}}) - \lambda \shD_{\text{3D}}^T \widetilde{\vec{v}}^{(i+1)} \big)
\end{array} \right.
\label{eq:pdfp}
\end{equation}
where $\widetilde{\vec{f}}^{(i)} \in \R^{NT}$ is the reconstruction, $\widetilde{\vec{v}}^{(i)} \in \R^{RNT}$ is its dual variable vector in the shearlet domain and $\widetilde{\vec{d}}^{(i)} \in \R^{NT}$ a temporary variable vector at iteration $i$.

A similar iteration scheme can be derived for the static case of equation~\eqref{eq:StaticVMdiscr} by substituting, for each $\tau = 1, \ldots, T$, $\widetilde{\RadonD} \in \R^{NT \times MT}$ with $\RadonD_{\tau}\in \R^{N \times M}$, $\widetilde{\vec{f}}^{(k)} \in \R^{NT}$ with $\vec{f}_{\tau}^{(k)} \in \R^N$ and $\widetilde{\vec{y^{\eta}}} \in \R^{MT}$ with $\vec{y}_{\tau}^{\vec{\eta}} \in \R^M$. As a consequence the temporary vector $\vec{d}^{(k)}$ belongs to $\R^N$ and the dual vector $\vec{v}^{(k)}$ belongs to the 2D shearlet or wavelet domain.

Observe that, even though the PDFP algorithm changes the role of the parameter $\alpha$ from the weight of the regularization term to the thresholding value when the dual iterate $\widetilde{\vec{v}}^{(k)}$ in~\eqref{eq:pdfp} is computed, it does not eliminate the need of picking a suitable value for $\alpha$. This is a notoriously difficult task with many options available. Which one to pick depends on the problem at hand. Here, we apply an automated controlled sparsity scheme, called controlled wavelet domain sparsity (CWDS), first proposed in~\cite{purisha2017}.

\begin{algorithm}[t] \label{al:CWDS}
\SetKwInOut{Input}{input}
\SetAlgoLined
\caption{PDFP with controlled sparsity (CWDS-PDFP) for equation~\eqref{eq:pdfp}}
\Input{measurements $\widetilde{\vec{y^{\eta}}}$, forward operator $\widetilde{\RadonD}$, parameters $\gamma, \lambda > 0$, \emph{a priori} sparsity $C_\text{pr}$, iteration limit $\mathcal{I}$, tolerances $\delta_1, \delta_2$ and initialization parameters $\zeta, \omega > 0$.}
\vspace{2mm}\textbf{initialize}: $\widetilde{\vec{f}}^{(0)} = \vec{0}, \widetilde{\vec{v}}^{(0)} = \vec{0}, k = 0, e^{(0)} = 1$ and $C^{(0)} = 1$\;
\textit{\% calculate initial thresholding parameter} \\
$h=(1-C_\text{pr})R$  \;
$\alpha^{(0)} = \zeta \bigtriangleup_h \big( \shD_{\text{3D}} \widetilde{\RadonD}^T\widetilde{\vec{y}^{\eta}} \big)$ \;
\textit{\% calculate initial step length} \\
$\beta \ = \ \omega \alpha^{(0)}$ \;
\vspace{2mm}\While{$i < \mathcal{I}$, and $|e^{(i)}| \geq \delta_1$ or $\epsilon \geq \delta_2$}{
 \vspace{2mm}\textit{\% check difference in sparsity} \\
  $e^{(i+1)} = C^{(i)} - C_\text{pr}$\;
 \vspace{2mm}\textit{\% shrink the step size if necessary} \\
  \If{$\sgn(e^{(i+1)}) \neq \sgn(e^{(i)})$}{
    $\beta = \beta(1 - |e^{(i)} - e^{(i-1)}|)$\;}
 \vspace{2mm}\textit{\% update thresholding value} \\
  $\alpha^{(i+1)} = \max\lbrace 0, \alpha^{(i)} + \beta e^{(i+1)} \rbrace$\;
  \vspace{2mm}\textit{\% perform PDFP step} \\
  $\widetilde{\vec{d}}^{(i+1)} = \max\lbrace \vec{0}, \widetilde{\vec{f}}^{(i)} - \gamma (\widetilde{\RadonD}^T \widetilde{\RadonD}\widetilde{\vec{f}}^{(i)} - \widetilde{\RadonD}^T\widetilde{\vec{y}^{\eta}}) - \lambda \shD_{\text{3D}}^T \widetilde{\vec{v}}^{(i)} )\rbrace$\;
  $\widetilde{\vec{v}}^{(i+1)} =  (\mathbbm{1} - S_{\alpha^{(i)}\frac{\gamma}{\lambda}})(\shD_{\text{3D}} \widetilde{\vec{d}}^{(i+1)} + \widetilde{\vec{v}}^{(i)})$\;
  $\widetilde{\vec{f}}^{(i+1)} = \max\lbrace \vec{0}, \widetilde{\vec{f}}^{(i)} - \gamma (\widetilde{\RadonD}^T \widetilde{\RadonD}\widetilde{\vec{f}}^{(i)} - \widetilde{\RadonD}^T\widetilde{\vec{y}^{\eta}}) - \lambda \shD_{\text{3D}}^T \widetilde{\vec{v}}^{(i+1)} )\rbrace$\;
  \vspace{2mm}\textit{\% update current sparsity level} \\
  $C^{(i+1)} = \frac{\#_{\kappa} (\SH_{\text{3D}} \widetilde{\vec{f}}^{(i+1)})}{RNT}$\;
  \vspace{2mm}\textit{\% check relative change of the reconstruction} \\
  $\epsilon = \frac{\|\widetilde{\vec{f}}^{(i+1)} - \widetilde{\vec{f}}^{(i)}\|_2}{\|\widetilde{\vec{f}}^{(i+1)}\|_2}$\;
  \vspace{2mm}\textit{\% update iteration counter} \\
  $i = i+1$\;
}
\textbf{return}
$\widetilde{\vec{f}}^{(i)}$\;
\end{algorithm}

In CWDS, the level of sparsity of the reconstruction is driven, through a control tuning algorithm, to an \textit{a priori} known ratio of nonzero versus zero wavelet or shearlet coefficients in the unknown.
The \emph{a priori} sparsity of the unknown is computed either from the groundtruth, when available, or by estimating it from high resolution reconstructions, such as FBP from densely sampled angles. 
The \textit{a priori} sparsity level is calculated according to the formula:
\[
C_\text{pr} = \frac{\#_{\kappa} \big(\shD_{\text{3D}} \widetilde{\vec{f}}\, \big)}{RNT},
\]
where $\#_{\kappa}$ counts the number of coefficients with absolute value larger than $\kappa$ and $R$ is the total number of shearlet coefficients for the given sized target.
Then, at each iteration $k=0,1,\ldots$, the sparsity level $C^{(k)}$ of the current iterate $\widetilde{\vec{f}}^{(k)}$ is calculated and compared to $C_\text{pr}$. If the current sparsity is too high, more coefficients need to be eliminated and hence $\alpha^{(k)}$ is updated to a smaller value, and vice versa. As the sparsity levels get closer, the tuning gets finer and finer. In this way, the changes can initially be fast but later oscillating sparsity levels are avoided. The initial thresholding parameter $\alpha^{(0)}$ is chosen as in the original paper~\cite{purisha2017}:
\[
\alpha^{(0)} = \zeta \bigtriangleup_h  \big(\shD_\text{3D} \widetilde{\RadonD}^T \widetilde{\vec{y^{\eta}}} \big),
\]
where $\bigtriangleup_h \vec{x}$ denotes the mean of the $h=\lceil(1-C_\text{pr})R\rceil$ biggest components of a given vector $\vec{x}$ in absolute value.

The pseudocode of PDFP with controlled sparsity, for the dynamic case, is summarized in Algorithm~\ref{al:CWDS}. The pseudocode reads similarly for the static case, with the suitable adjustments of the PDFP iteration, as described above.

\section{Experiments and results}
\label{sec:ExpAndRes}
In this section, we evaluate the performance of the proposed motion-aware reconstruction approach. For thorough testing we consider a combination of different measurement setups and different kinds of simulated and measured data.

\subsection{Preliminaries}
Let us begin by describing the considered experimental scenarios and giving details on the implementation of the used operators, the methods that we compare our results with and the metrics we use to compare the results.

\subsubsection{Experimental scenarios} 
\label{sec:data}
We consider three different types of data for evaluating our proposed method: a digital data set and two measured data.  
A common characteristic to all of the data sets is a static stem-like structure containing a dynamic contrast agent which changes in visibility across the time steps.
The digital phantom was specifically created to simulate the increasing intensity of the contrast agent and other relevant features that are know to be present in the real data. The measured data are the results of X-ray microtomography measurements of two targets: a physical gel phantom and a young living \emph{populus tremula} tree grown in laboratory from \emph{in-vitro} propagated material. 
\begin{figure}[t]
    \centering
    \includegraphics[width=0.35\columnwidth]{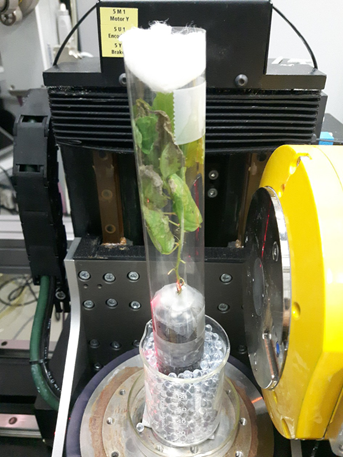}
    \caption{Living \emph{Populus tremula} ready to be imaged in the micro-CT scanner at University of Helsinki micro-CT laboratory. The plant was first placed in 50 ml Falcon tubes filled with peat soil, supported with a tube made of clear plastic film and covered with a moist tissue to reduce drying during imaging.}
    \label{fig:livingTree}
\end{figure}

\begin{enumerate}
\item \textbf{Digital plant stem phantom:} The simulated data consist of a $256 \times 256 \times 34$ array of attenuation coefficients where the third dimension serves as time. The background is a static image simulating the stem of a plant and, on top of it, we imposed five spreading points with higher attenuation values to simulate the spreading contrast agent. Figure \ref{fig:plantphantom} illustrates four different time steps of the simulated phantom.
\begin{figure}[t]
\centering
\includegraphics[width=\textwidth]{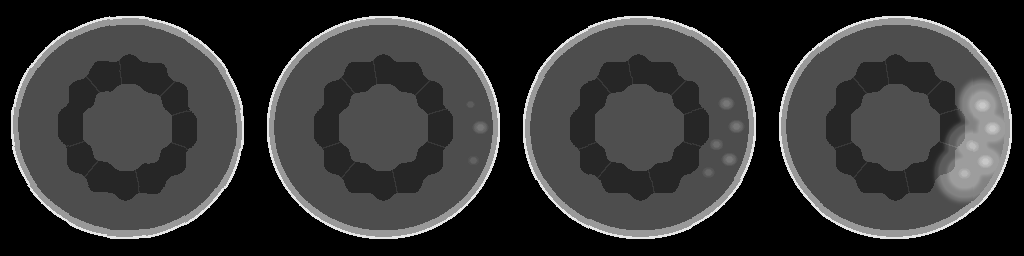}
\caption{Time instances $\tau$ = 1, 10, 16 and 34 of the $256 \times 256 \times 34$ digital plant stem phantom. The evolution across the time frames simulates contrast agent spreading in the plant stem.} \label{fig:plantphantom}
\end{figure}
Noisy measurements are simulated with ASTRA~\cite{van2016} using a fanbeam geometry: the ASTRA operator is applied independently to each time frame of the phantom, considering the same sampled angles for each time frame. To avoid inverse crime, measurements are simulated at a higher resolution and then downsampled or \textit{binned}, namely, two neighbouring detector elements are considered as a single wider element providing just a single measurement. Noise is simulated using additive Gaussian noise with $0$ mean and $1 \%$ variance. Due to the noise reducing effect of binning, the Gaussian noise was added afterwards to keep the noise at a known level.

\item \textbf{Gel phantom:}
The physical phantom is made of agarose gel in order to demonstrate iodide diffusion and geometry prior to imaging living samples, which cannot withstand high radiation doses from a denser set of measurement angles. 2\% agarose gel was cast into 50 ml Falcon test tubes and let set at room temperature. Once solidified, 5 intact plastic straws were inserted into the gel to create cavities by removing the gel in the tubes. These gel cavities were filled with 20\% sucrose solution to improve the diffusion capability of iodide into the gel matrix (to direct osmosis from the cavities to the gel matrix). After this, plastic straws that were pierced densely with a needle were inserted into the sucrose solution filled cavities.
All air bubbles were removed to assure proper diffusion. These straws acted both as physical barriers between the liquid phase and the gel matrix to slow down iodine diffusion rate and, more importantly, the pierced holes mimicked the function of plasmodesmata (which act as cell-to-cell passages between, \eg{}, phloem sieve elements or xylem tracheary elements and surrounding tissues).

After the first round of imaging, the gel phantom was manually diffused with a sugar-iodine-mix (potassium iodide diluted with 20\% sucrose solution for a 30\% concentration).  The first iodide application was given to the centermost cavity, the second application to the centermost and neighboring cavities, the third application to all 5 cavities and so on. An increasing amount of iodide was applied to the cavities between measurements.
The idea was to mimic an initial single cell spot exposure pattern that expands gradually laterally to neighboring cells with increasing iodide load (such a pattern would be a sum of both vertical diffusion from above via interconnected phloem sieve element or xylem tracheary element cells, and lateral cell-to-cell diffusion via plasmodesmata). 
Approximately 5 ml of sugar-iodine mixture was used in total and the phantom was imaged 17 times with 360 measurement angles. Each imaging round lasted 4.5 minutes and the interval between each imaging round was approximately 10 minutes. In total the whole experiment was conducted in 3 hours.

\item \textbf{Living tree:}
The \emph{in vitro} grown (\emph{Populus tremula}) plants were first placed in 50 ml Falcon tubes filled with peat soil, supported with a tube made of clear plastic film and covered with a moist tissue (to reduce drying during imaging) and let acclimate to imaging room temperature and humidity (Figure~\ref{fig:plantphantom}). Once set, the plant was passively perfused with a liquid sugar-iodine-mix (Iopamidol diluted with 20\% sugar water to lower the Iopamidol concentration to 25\%~\cite{keyes2017}) through cut leaf petiole. First round of imaging was done with 720 projections, followed by 11 imaging rounds with just 90 measurement angles. At the last time step the plant was imaged with 1440 angles.

The total amount of Iopamidol was approximately 250 $\upmu$l (as plant uptake and evaporation from the test tube during imaging could not be measured). Each set of 90 projections took 3 minutes to measure and the interval between time steps was approximately 15 minutes. The overall experiment time for the data used here was 4 hours. 
\end{enumerate}

The measurements for both the gel phantom and \textit{in vivo} plant were done at the University of Helsinki microCT laboratory with a Phoenix Nanotom S scanner~\cite{suuronen2014}. For the physical gel phantom the exposure time was 500 ms with 250 ms pause for rotations, X-ray tube acceleration voltage was 60 kV and tube current 250~$\upmu$A. For the living plant the exposure time was 1000~ms with 1000~ms pause for rotations, X-ray tube acceleration voltage 45 kV and tube current 335~$\upmu$A.

In both the gel phantom and \textit{in vivo} plant case the first and last set of measurements were originally done with a noticeably higher number of projections in order to obtain a good quality reconstruction for comparison. Both targets were imaged with 720 projections before the contrast agent was applied and the final measurement of the living plant, when the iodine has fully perfused into the stem, consists of 1440 projections. However, during the reconstruction phase, the measurements for these frames were downsampled to a lower projection count by picking fewer angles when the sinograms were constructed: for simplicity, the number of measurement angles was kept uniform and equal to the other frames.

\subsubsection{Operators}
All the forward operators were created in Matlab using the Spot operator~\cite{spot2012} from ASTRA~\cite{van2016}: Spot operators provide matrix-like memory efficient computations, compared to the standard matrix implementation of ASTRA. In particular, for each $\tau=1,\ldots,T$, each $\RadonD_{\tau}$ is built sampling the \textit{same} angles. While, in principle, sampling different angles at different time instances might provide more information than sampling always the same ones, in practice we observed a negligible difference which was not worth the effort of building different forward operators for each time instance.

For all our experiments involving shearlets, we are using compactly supported shearlets provided by the Shearlab 3D toolbox~\cite{Kutyniok2016} for Matlab. The 2D shearlet transform $\shD_\text{2D}$ was computed with 3 scales while the 3D shearlet transform $\shD_\text{3D}$ with 2 scales, both with the default number of shearing directions, resulting in 33 and 99 subbands, respectively. For the wavelet transform $\mathbf{W}_{\text{2D}}$ we used the implementation of Haar wavelets with 4 scales provided by the Wavelet Toolbox~\cite{wavelettoolbox2018}. 

Finally, notice that the 3D shearlet transform used in the motion-aware model proposed requires that the target is at least 33 pixels long in each dimension. For this reason, both measured data sets were inflated in the temporal dimension by adding each frame multiple times to obtain 34 frames for the physical phantom and 33 frames for the real plant data set. As a result, the reconstruction algorithm produces multiple reconstructions from the same sinograms. These can be either kept separate or be combined into a single reconstruction. For the static models with 2D shearlets and 2D wavelets such steps were not needed.

\subsubsection{Compared Methods} 
We compare our results with a variety of methods, briefly summarized in the following list.
\begin{itemize}
\item \textbf{FBP:} standard filtered back-projection using a \emph{Ram-Lak} filter, as provided by ASTRA. Each time frame is reconstructed independently without any temporal connection between time frames.
\item \textbf{2D wavelets:} The $\ell_1$-regularized wavelet solution of~\eqref{eq:StaticVMdiscr} with $\mathbf{\Psi} = \mathbf{W}_{\text{2D}}$. The PDFP algorithm with controlled sparsity is implemented using 2D Haar wavelets, as in the original paper~\cite{purisha2017}. The uniform sparsity control step is performed on each time step separately and there is not any temporal connection between time frames.
\item \textbf{2D shearlets:} The static $\ell_1$-regularized shearlet solution of~\eqref{eq:StaticVMdiscr} with $\mathbf{\Psi} = \shD_{\text{2D}}$. The PDFP algorithm with controlled sparsity is implemented using 2D compactly supported shearlets. The uniform sparsity control step is performed on each time step separately and there is not any temporal connection between time frames.
\end{itemize}

For measured data, FBP reconstructions from dense set of projections were used also for error estimates and for computing the desired sparsity level $C_\text{pr}$. For the gel phantom each time frame has been reconstructed from 360 projections and for the real plant the first and last measurements were reconstructed from 720 and 1440 projections, respectively.

It is worth noting that with the 2D static methods the different time frames were not completely independent. While the PDFP step was done to each time frame separately, the sparsity level was updated based on the average sparsity level. However, when the initial sparsity levels were calculated, the differences between the time frames were minimal, in general less than the sparsity threshold. While this approach may cause non-optimal changes in the sparsity level for some time frames during the iteration, it should not have a noticeable effect on the convergence rate since the measurement data are relatively similar across time. If, for example, the number of projections between time steps was not uniform, a different approach could be more suitable.

\subsubsection{Parameters}
\begin{table}[b]
\caption{Fixed parameter values for the CWDS-PDFP algorithm~\ref{al:CWDS}.} 
\begin{tabular}{@{\;}lll@{\;}}
\specialrule{.1em}{.05em}{.05em} 
Parameter & Value & Notes \\ 
\hline
$\gamma$ & $1$ & See equation \eqref{eq:pdfp}. \\
$\lambda$ & $0.99$ & See equation \eqref{eq:pdfp}. \\
$\mathcal{I}$ & $300$ & Iteration limit. \\
$\delta_1$ & $0.01$ & Tolerance for sparsity level. \\
$\delta_2$ & $0.003$ & Tolerance for norm difference. \\
\specialrule{.1em}{.05em}{.05em} 
\end{tabular}
\label{tab:fixed}
\end{table}
The PDFP algorithm with controlled sparsity requires some pre-set parameters, most of which are kept fixed throughout the experiments. These are collected in table~\ref{tab:fixed}. Other parameters, instead, depend on the data and the deployed operators, such as the \emph{a priori} sparsity level $C_\text{pr}$ or the parameter $\omega$ used to tune the initial speed of the sparsity controller. In addition, there are the threshold $\kappa$ for calculating the sparsity level at each iteration and the parameter $\zeta$ used to initialize the thresholding value $\alpha^{(0)}$. For each data set, these varying parameters are shown in table \ref{tab:param}.

\begin{table}[ht]
\caption{Varying parameter values for the CWDS-PDFP algorithm~\ref{al:CWDS}.} 
\begin{tabular}{@{\;}llllll@{\;}}
\specialrule{.1em}{.05em}{.05em} 
Dataset & Method & $C_\text{pr}$ & $\omega$ & $\kappa$ & $\zeta$ \\ 
\hline
 & Haar wavelet & $0.30$ & $10$ & $10^{-6}$ & $1$ \\ 
Digital & 2D shearlet & $0.77$ & $50$ & $10^{-5}$ & $1$ \\
phantom & 3D shearlet & $0.73$ & $10$ & $10^{-6}$ & $1$ \\ \hline
 & Haar wavelet & $0.54$ & $0.5$ & $10^{-6}$ & $1$ \\
Gel & 2D shearlet & $0.82$ & $100$ & $10^{-6}$ & $1$ \\
phantom & 3D shearlet & $0.76$ & $10$ & $10^{-6}$ & $10$ \\ \hline
 & Haar wavelet & $0.45$ & $0.05$ & $10^{-6}$ & $2$ \\
Real plant & 2D shearlet & $0.94$ & $5$ & $10^{-5}$ & $10$ \\
 & 3D shearlet & $0.95$ & $10$ & $10^{-5}$ & $100$ \\
\specialrule{.1em}{.05em}{.05em} 
\end{tabular}
\label{tab:param}
\end{table}

\subsubsection{Similarity measures}
For the assessment of the image quality, we are using several quantitative metrics. Besides the $\ell_2$ relative error given by:
\[
\frac{\| \widetilde{\vec{f}} - \vec{f} \|_2}{\| \vec{f} \|_2}
\]
where $\vec{f}$ denotes the reference image and $\widetilde{\vec{f}}$ a reconstruction, we considered the peak signal-to-noise ratio (PSNR) and the Haar wavelet-based perceptual similarity index (HPSI), recently introduced in~\cite{reisenhofer2018}, which should give a more natural assessment of the similarity between two images. Nonetheless, we stress that, from the perspective of the end-user, the emphasis is almost purely on the visual quality, especially on how well the contrast agent can be spotted.

\subsection{Results}
\label{sec:results}
In the following, we present and discuss our numerical experiments. Computations were implemented with Matlab R2018a, running on a Windows 10 computer with 16GB of 2.40 GHz DDR4 memory and Intel i5 CPU at 2.80 GHz.

\subsubsection{Digital plant stem phantom}

Image quality metrics are reported in table~\ref{tab:brightsimu}. Notice that, unlike for the measured data where we use the dense FBP reconstruction as reference image, to compute the metrics in the simulated case we use the ground-truth. A visualization of the reconstruction quality for the last frame ($\tau=T=34$) and for different amount of projections ($P=45, 90$ and 360) is given in figure~\ref{fig:brightsimu}, where for better visualization we omitted the reconstructions from 120 projections.

\begin{table}[htb]
\caption{Error metrics of the digital phantom reconstructions. Each frame is compared against the known ground-truth. Reconstructions are shown in figure~\ref{fig:brightsimu}.} \label{tab:brightsimu}
\begin{tabular}{@{}llllll}
\specialrule{.1em}{.05em}{.05em} 
Method & $P$ & $\ell^2$-error & PSNR & HPSI \\ 
\hline
FBP & 45 & 38.6\% & 18.6 & 0.278 \\
 & 90 & 30.5\% & 20.7 & 0.433 \\
 & 120 & 27.4\% & 21.6 & 0.478 \\
 & 360 & 21.9\% & 23.6 & 0.575 \\ \hline
Haar CWDS-PDFP & 45 & \HL{26.5\%} & \HL{21.9} & \HL{0.368} \\
 & 90 & 22.4\% & 23.3 & 0.482 \\
 & 120 & 23.5\% & 22.9 & 0.520 \\
 & 360 & 23.9\% & 22.8 & 0.584 \\ \hline
SH$_\text{2D}$ CWDS-PDFP & 45 & 33.0\% & 20.0 & 0.366 \\
 & 90 & 30.6\% & 20.7 & 0.447 \\
 & 120 & 23.1\% & 23.1 & \HL{0.548} \\
 & 360 & 24.7\% & 22.5 & 0.584 \\ \hline
SH$_\text{3D}$ CWDS-PDFP & 45 & 28.7\% & 21.2 & 0.366 \\
 & 90 & \HL{23.6\%} & \HL{22.9} & \HL{0.500} \\
 & 120 & \HL{22.8\%} & \HL{23.2} & 0.541 \\
 & 360 & \HL{21.4\%} & \HL{23.7} & \HL{0.610} \\
\specialrule{.1em}{.05em}{.05em} 
\end{tabular}
\end{table}

\begin{figure}[!ht]
\begin{tabular}{r@{\hskip 3pt}ccc}
& 45 projections & 90 projections & 360 projections \\
\rotatebox[origin=l]{90}{\hspace{0.12\columnwidth} FBP}
& \includegraphics[width=0.3\columnwidth]{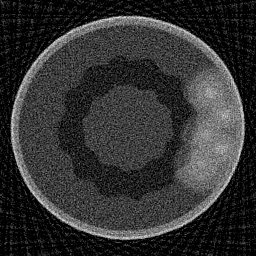}
& \includegraphics[width=0.3\columnwidth]{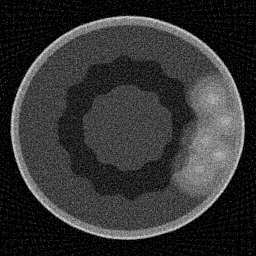}
& \includegraphics[width=0.3\columnwidth]{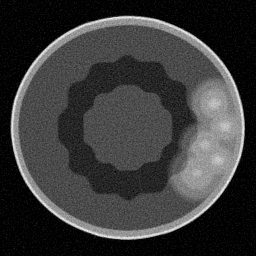} \\[6pt]
\rotatebox[origin=l]{90}{\hspace{0.08\columnwidth} Haar wavelet}
& \includegraphics[width=0.3\columnwidth]{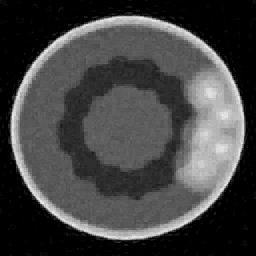}
& \includegraphics[width=0.3\columnwidth]{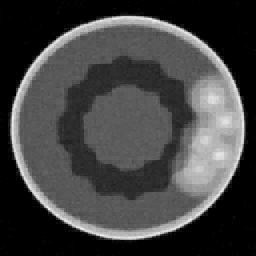}
& \includegraphics[width=0.3\columnwidth]{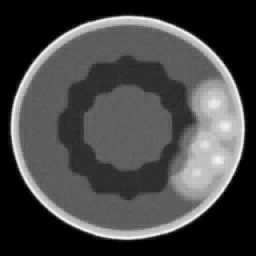} \\[6pt]
\rotatebox[origin=l]{90}{\hspace{0.09\columnwidth} 2D shearlet}
&\includegraphics[width=0.3\columnwidth]{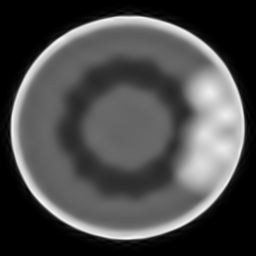}
&\includegraphics[width=0.3\columnwidth]{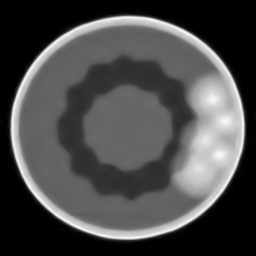}
&\includegraphics[width=0.3\columnwidth]{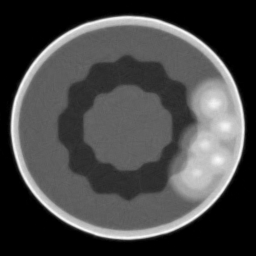} \\[6pt]
\rotatebox[origin=l]{90}{\hspace{0.09\columnwidth} 3D shearlet}
&\includegraphics[width=0.3\columnwidth]{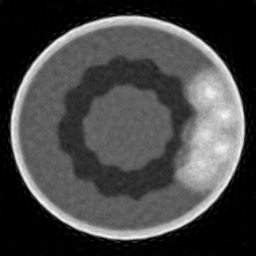}
&\includegraphics[width=0.3\columnwidth]{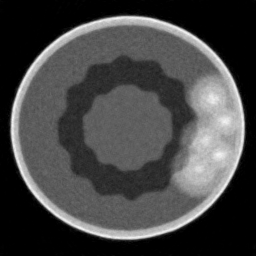}
&\includegraphics[width=0.3\columnwidth]{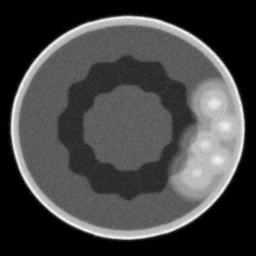}
\end{tabular}
\caption{Digital phantom reconstructed with different methods, for time frame $T = 34$. On each row, reconstructions from $P=45$ (leftmost), $P=90$ (middle) and $P=360$ (rightmost).} \label{fig:brightsimu}
\end{figure}

Compared to measured data, simulated data were designed to have significantly brighter contrast agent: as it is evident in figure~\ref{fig:brightsimu}, all the approaches provide clear reconstructions of the static structure and pick up the information on the spreading iodine. The decrease in quality with the lowest number of projections ($P=45$, rightmost column) is barely noticeable. With the highest number of projections ($P=360$, leftmost column) all methods provide clear reconstructions: the FBP reconstruction is quite noisy, but the iodine is clearly visible. With less projections ($P=90$, middle column) the static 2D shearlet approach provides blurry reconstructions and in the FBP reconstructions the noise becomes even more evident: the iodine is visible, but the spreading points are not easy to spot against the background. The static 2D wavelets and the proposed dynamic approach with 3D shearlets provide comparable reconstructions, even though the proposed approach is less jagged.

In the lowest number of projections case ($P=45$, rightmost column), FBP is clearly too noisy and the static 2D shearlet too blurry, even though the iodine is still clearly visible. Static 2D wavelets provide better results metrics-wise but with the proposed approach the blockiness is less prominent, which makes it visually more appealing especially where the iodine is spreading. Similar considerations can be done on the other time frames reconstructions, which are not reported here for brevity.

\subsubsection{Gel phantom}
Image quality metrics are reported in table~\ref{tab:5pt}, while number of iterations and computing times for each reconstruction method are included in table~\ref{tab:timetable}. 
A visualization of the reconstruction quality for different time frames ($\tau=2, 7$ and $12$) is given in figures~\ref{fig:5pt90angles} and~\ref{fig:5pt30angles} corresponding to $P=90$ and $P=30$ projection views, respectively. After downsampling of the data, each time frame has size $256 \times 256$. In both figures, the yellow color marks the iodine. Also, the first row shows the FBP reconstructions obtained from dense sampling ($P=360$): these were used to compute the quantitative metrics and as a visual benchmark. The ``doubled'' reconstructions coming from the expanded temporal dimension with the proposed 3D shearlet approach were combined together by taking their average. 

Compared to the simulated data, the physical phantom is more challenging: while the contrast agent is easily seen in the areas initially filled with sucrose, the spreading into the gel body is more intricate. Especially with a low number of projections ($P=30$) the bright yellow parts get smeared, indicating a greater amount of iodine than actually present. Indeed, the reconstructions have more background noise and especially the later frames have a lot of artifacts making it hard to differentiate between errors and actual contrast agent outside of the sucrose holes.

\begin{table}[h]
\caption{Error metrics of the gel phantom. Each frame is compared against the dense angle FBP reconstruction. Reconstructions are shown in figures~\ref{fig:5pt90angles} and~\ref{fig:5pt30angles}.} \label{tab:5pt}
\begin{tabular}{@{}llllll}
\specialrule{.1em}{.05em}{.05em} 
Method & $P$ & $\ell^2$-error & PSNR & HPSI \\ 
\hline
Haar CWDS-PDFP & 30 & 15.3\% & 23.4 & \HL{0.535} \\
 & 45 & 15.3\% & 23.4 & \HL{0.592} \\
 & 90 & \HL{14.4\%} & \HL{23.9} & \HL{0.614} \\
 & 120 & \HL{14.2\%} & \HL{24.1} & \HL{0.610} \\
 & 360 & \HL{14.5\%} & \HL{23.9} & \HL{0.530} \\ \hline
SH$_\text{2D}$ CWDS-PDFP & 30 & 15.5\% & 23.3 & 0.408 \\
 & 45 & 15.3\% & 23.4 & 0.423 \\
 & 90 & 15.1\% & 23.5 & 0.444 \\
 & 120 & 15.0\% & 23.6 & 0.453 \\
 & 360 & 14.8\% & 23.7 & 0.488 \\ \hline
SH$_\text{3D}$ CWDS-PDFP & 30 & \HL{15.0\%} & \HL{23.6} & 0.488 \\
 & 45 & \HL{14.9\%} & \HL{23.6} & 0.516 \\
 & 90 & 14.7\% & 23.7 & 0.525 \\
 & 120 & 14.7\% & 23.7 & 0.512 \\
 & 360 & 14.6\% & 23.8 & 0.525 \\
\specialrule{.1em}{.05em}{.05em} 
\end{tabular}
\end{table}

\begin{figure}[!htb]
\vspace{-3pt}
\begin{tabular}{r@{\hskip 3pt}ccc@{\hskip 3pt}l}
 & $\tau = 2 \sim 23$ min & $\tau = 7 \sim 72$ min & $\tau = 12 \sim 121$ min & \\

\rotatebox[origin=l]{90}{\hspace{0.12\columnwidth} FBP}
& \includegraphics[width=0.3\columnwidth]{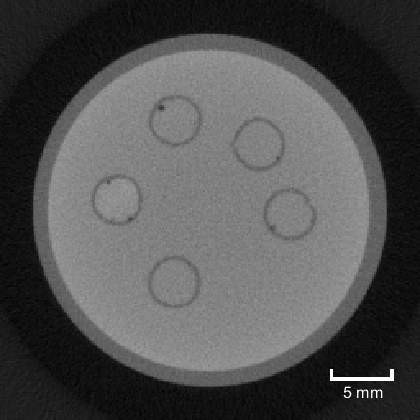}
& \includegraphics[width=0.3\columnwidth]{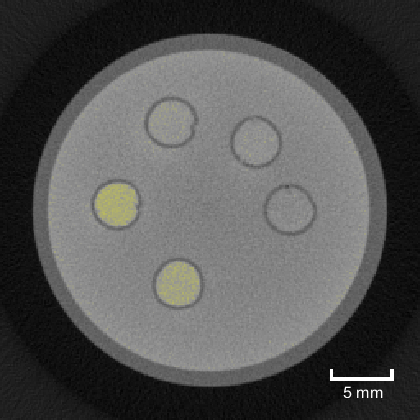}
& \includegraphics[width=0.3\columnwidth]{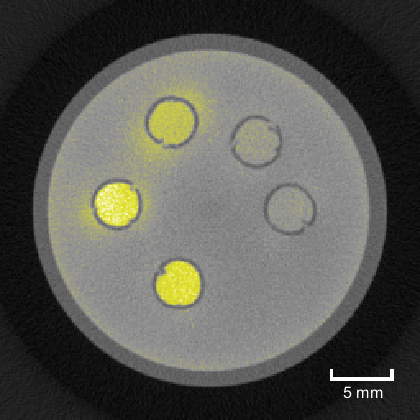}
& \includegraphics[height=0.30\columnwidth]{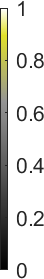} \\[-2pt]
& (a) & (b) & (c) & \\[3pt]
\rotatebox[origin=l]{90}{\hspace{0.08\columnwidth} Haar wavelet}
& \includegraphics[width=0.3\columnwidth]{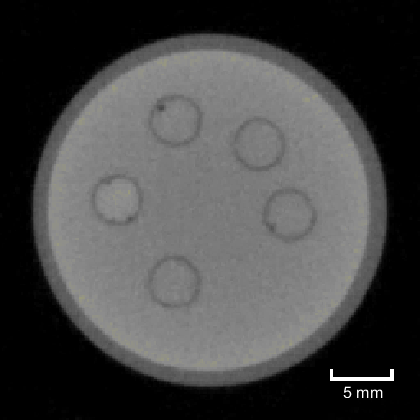}
& \includegraphics[width=0.3\columnwidth]{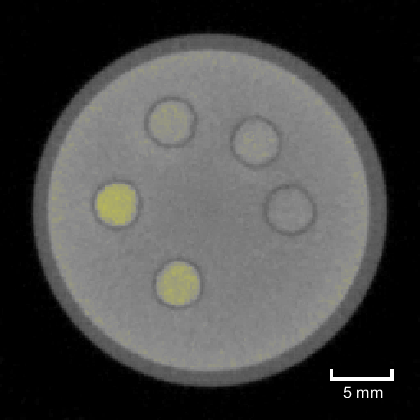}
& \includegraphics[width=0.3\columnwidth]{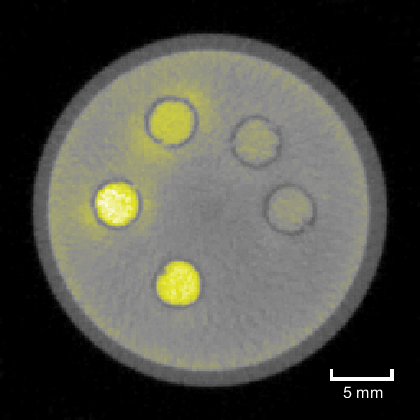}
& \includegraphics[height=0.30\columnwidth]{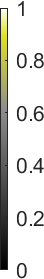} \\[-2pt]
& (d) & (e) & (f) & \\[3pt]
\rotatebox[origin=l]{90}{\hspace{0.09\columnwidth} 2D shearlet}
& \includegraphics[width=0.3\columnwidth]{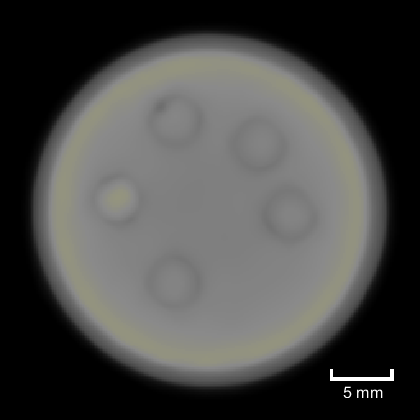}
& \includegraphics[width=0.3\columnwidth]{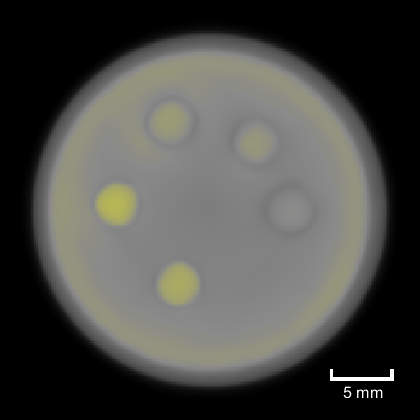}
& \includegraphics[width=0.3\columnwidth]{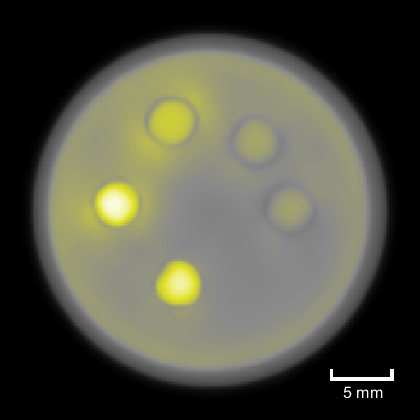}
& \includegraphics[height=0.30\columnwidth]{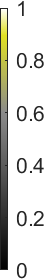} \\[-2pt]
& (e) & (g) & (h) & \\[3pt]
\rotatebox[origin=l]{90}{\hspace{0.09\columnwidth} 3D shearlet}
& \includegraphics[width=0.3\columnwidth]{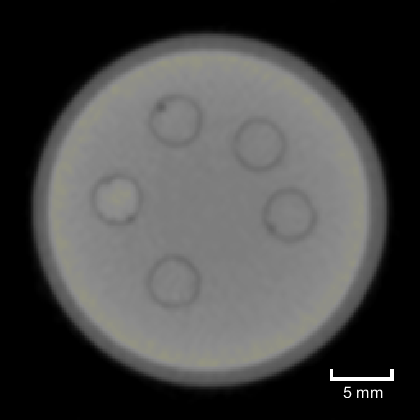}
& \includegraphics[width=0.3\columnwidth]{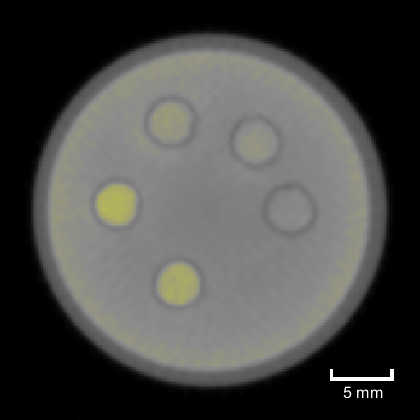}
& \includegraphics[width=0.3\columnwidth]{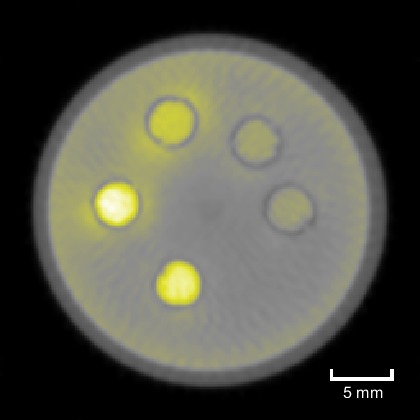}
& \includegraphics[height=0.30\columnwidth]{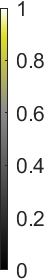} \\[-2pt]
& (i) & (j) & (k) & 
\end{tabular}
\vspace{-14.3pt}
\caption{Gel phantom reconstructed with different methods. First row $\big($(a)-(c)$\big)$: benchmark FBP reconstructions from $P = 360$. Other rows $\big($(d)-(l)$\big)$: $P = 90$. On each row, time frames $\tau=2$ (leftmost), $\tau=7$ (middle) and $\tau=12$ (rightmost).} 
\label{fig:5pt90angles}
\end{figure}

\begin{figure}[!htb]
\vspace{-3pt}
\begin{tabular}{r@{\hskip 3pt}ccc@{\hskip 3pt}l}
 & $\tau = 2 \sim 23$ min & $\tau = 7 \sim 72$ min & $\tau = 12 \sim 121$ min & \\

\rotatebox[origin=l]{90}{\hspace{0.12\columnwidth} FBP}
& \includegraphics[width=0.3\columnwidth]{figure3j.png}
& \includegraphics[width=0.3\columnwidth]{figure3k.png}
& \includegraphics[width=0.3\columnwidth]{figure3l.png}
& \includegraphics[height=0.30\columnwidth]{figure3l2.png} \\[-2pt]
& (a) & (b) & (c) & \\[3pt]
\rotatebox[origin=l]{90}{\hspace{0.08\columnwidth} Haar wavelet}
& \includegraphics[width=0.3\columnwidth]{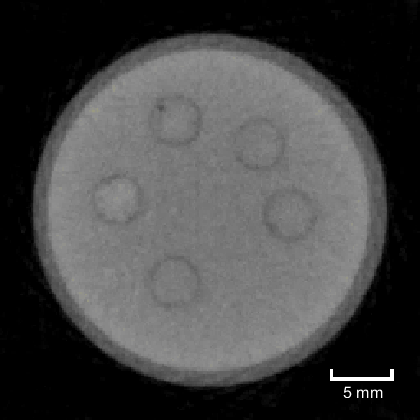}
& \includegraphics[width=0.3\columnwidth]{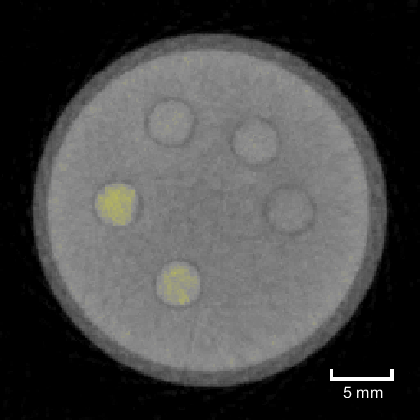}
& \includegraphics[width=0.3\columnwidth]{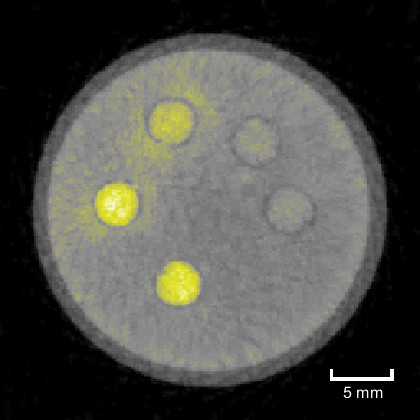}
& \includegraphics[height=0.30\columnwidth]{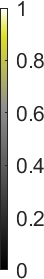} \\[-2pt]
& (d) & (e) & (f) & \\[3pt]
\rotatebox[origin=l]{90}{\hspace{0.09\columnwidth} 2D shearlet}
& \includegraphics[width=0.3\columnwidth]{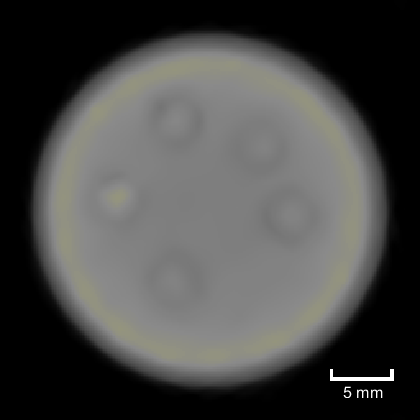}
& \includegraphics[width=0.3\columnwidth]{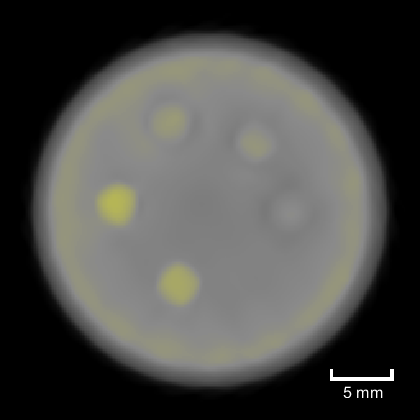}
& \includegraphics[width=0.3\columnwidth]{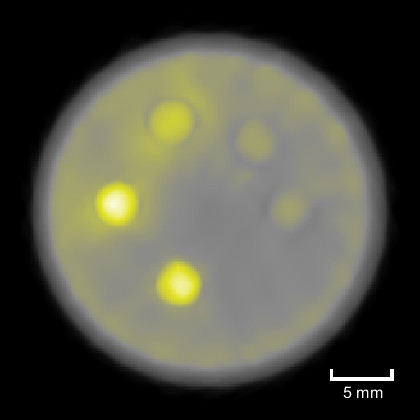}
& \includegraphics[height=0.30\columnwidth]{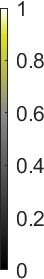} \\[-2pt]
& (e) & (g) & (h) & \\[3pt]
\rotatebox[origin=l]{90}{\hspace{0.09\columnwidth} 3D shearlet}
& \includegraphics[width=0.3\columnwidth]{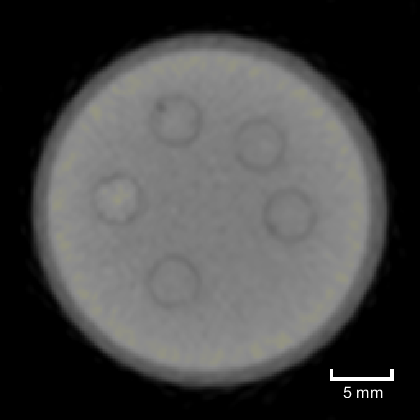}
& \includegraphics[width=0.3\columnwidth]{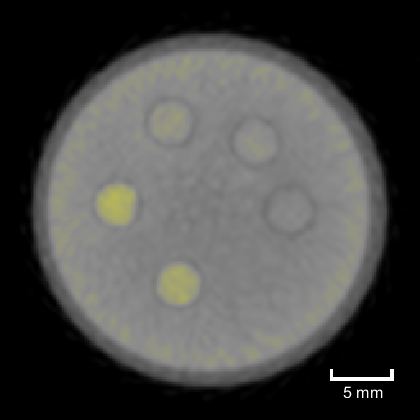}
& \includegraphics[width=0.3\columnwidth]{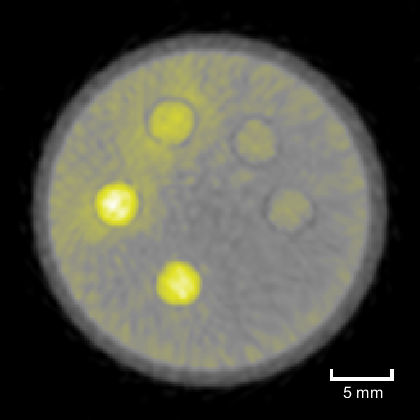}
& \includegraphics[height=0.30\columnwidth]{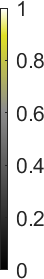} \\[-2pt]
& (i) & (j) & (k) & 
\end{tabular}
\vspace{-14.3pt}
\caption{Gel phantom reconstructed with different methods. First row $\big($(a)-(c)$\big)$: benchmark FBP reconstructions from $P = 360$. Other rows $\big($(d)-(l)$\big)$: $P = 30$. On each row, time frames $\tau=2$ (leftmost), $\tau=7$ (middle) and $\tau=12$ (rightmost).} 
\label{fig:5pt30angles}
\end{figure}

\clearpage
\begin{table}[htb]
\caption{Iterations and computing time required for each reconstruction method for the gel phantom. Reconstructions are shown in figures~\ref{fig:5pt90angles} and~\ref{fig:5pt30angles}.} \label{tab:timetable}
\begin{tabular}{@{}llllll}
\specialrule{.1em}{.05em}{.05em} 
Method & $T$ & $P$ & \begin{tabular}{@{}l@{}}Total \\ time (s)\end{tabular} & \begin{tabular}{@{}l@{}}Number of \\ iterations\end{tabular} & \begin{tabular}{@{}l@{}}Time per \\ iteration (s)\end{tabular} \\ 
\hline 
FBP & 17 & 360 & 14 & \ -- & \ -- \\ \hline
Haar CWDS-PDFP & 17 & 30 & 323 & 119 & 2.7 \\
 & & 45 & 464 & 116 & 4.0 \\
 & & 90 & 753 & 106 & 7.1 \\
 & & 120 & 917 & 99 & 9.3 \\
 & & 360 & 1434 & 53 & 26.9 \\ \hline
SH$_\text{2D}$ CWDS-PDFP & 17 & 30 & 379 & 45 & 8.4 \\
 & & 45 & 438 & 43 & 10.2 \\
 & & 90 & 582 & 45 & 12.9 \\
 & & 120 & 807 & 45 & 17.9 \\
 & & 360 & 1846 & 55 & 33.6 \\ \hline
SH$_\text{3D}$ CWDS-PDFP & 34 & 30 & 3100 & 53 & 58.5 \\
 & & 45 & 3510 & 55 & 63.8 \\
 & & 90 & 3630 & 56 & 64.8 \\
 & & 120 & 3554 & 50 & 71.1 \\
 & & 360 & 5689 & 54 & 105.3 \\
\specialrule{.1em}{.05em}{.05em} 
\end{tabular}
\end{table}

With $P=90$ projections, both static 2D and dynamic 3D shearlet methods produce reconstructions with less background noise, compared to FBP and static 2D wavelets, making it slightly easier to notice the leaking iodine. Between the two, static 2D shearlets provide a somewhat more misleading reconstruction since yellow, even though faded, is more evident in the gel body than with the proposed 3D shearlet approach. Beside static 2D shearlets, the other methods produce relatively little noise and, even with the highest amount of contrast agent, the artifacts are manageable. 

Overall, static 2D wavelets and the proposed dynamic 3D shearlets yields similar visual information about the spreading of iodine, while static 2D shearlets provide clearly poorer reconstructions given the general blurriness and the iodine smearing. Metrics-wise the proposed approach performs better with the least projections ($P=30, 45$), which is of most interest for the end-user since it results in a lesser amount of X-ray radiation for the target and lower scanning time.

\subsubsection{Living tree data}
In the experiments described in this section, we test the performance on the proposed method for data on an \textit{in vivo} -imaged plant. Compared to the previous ones, this is a remarkably more difficult task: the target is smaller, the concentration of iodine is much smaller and, most importantly, the longer time required for scanning the denser data for benchmark visualization might result in morphological changes in the sample owing to radiation damage. This means that the finer details of structures or even positions of individual cells can change over time, which was not the case for the gel phantom.

Image quality metrics are reported in table~\ref{tab:plant}, while the number of iterations and computing times of the full set of reconstructions for each method are included in table~\ref{tab:plantTimetable}. Higher quality FBP reconstructions from dense sampling are shown in figure~\ref{fig:plantFBP}: on the left, the reconstruction from the initial time frame ($P = 720$); on the right, the reconstruction of the last frames ($P = 1440$). These were used for visual comparison and for computing the mean quantitative metrics of the first and last frame reconstructed with the other methods. The values reported in table~\ref{tab:plant} are the mean of the metrics computed for the first and last frame with respect to the other methods.

A visualization of the reconstruction quality for different time frames ($\tau=1, 7$ and $11$) is given in figures~\ref{fig:plant90angles} and~\ref{fig:plant30angles} corresponding to $P=90$ and $P=30$ projection views, respectively. After downsampling of the data, each time frame has size $155 \times 155$, but the final reconstructions were cropped to size $128 \times 128$.  From the ``tripled'' reconstructions coming from the expanded temporal dimension with the proposed 3D shearlet approach, we chose the middle one for display and the rest were discarded. Combining three noisy images would have amplified the reconstructions errors where as with the two images of the gel phantom the difference was minor. In all the figures, the yellow color marks the iodine.

The real plant is a small target with relatively low concentration of contrast agent. This makes it difficult to obtain clear reconstructions: even with very dense projection sample, the background noise is extremely noticeable. Compared to the other data sets, the static 2D wavelet approach provides a poorer performance: while for the simulated data and the gel phantom reconstructions were comparable to the proposed approach, in this case the noise is dominant and even in the last frame it looks like there is almost no iodine. This is evident already with $P=90$ projections. The static 2D shearlet approach suffers from the same limitations observed before: reconstruction are too blurry, completely smoothing out the edges. As a result, the iodine smears to regions where it should not be present, as a comparison with the high resolution reconstructions in figure~\ref{fig:plantFBP} shows.

The proposed approach, instead, correctly displays the right amount of iodine in the right spots. This is confirmed also by the quantitative metrics: for $P=90$ projections, the proposed approach outperforms the others and for $P=30$ the values are quite close to the ones of static 2D shearlets, which in contrast give a very poor reconstruction.  Keeping in mind that all the challenges that the living tree data set poses, the reconstruction quality with the proposed 3D shearlet approach is noteworthy.

\begin{figure}[!htb]
\begin{tabular}{r@{\hskip 3pt}cc@{\hskip 3pt}l}
& $\tau = 1 \sim 0$ min & $\tau = 11 \sim 235$ min & \\
\rotatebox[origin=l]{90}{\hspace{0.19\columnwidth} FBP}
& \includegraphics[width=0.45\columnwidth]{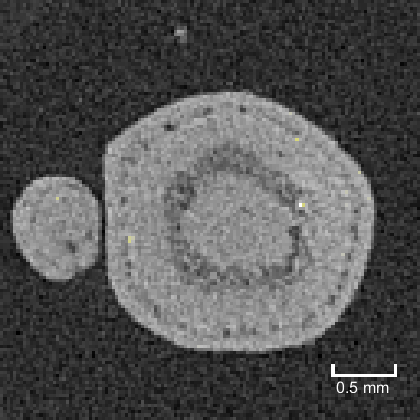}
& \includegraphics[width=0.45\columnwidth]{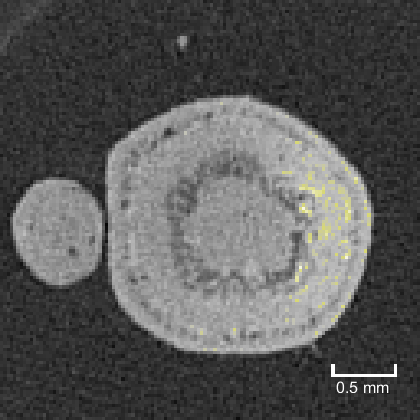}
& \includegraphics[height=0.45\columnwidth]{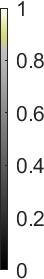} \\[-2pt]
& (a) & (b) &
\end{tabular}
\caption{The FBP reconstructions of the living tree used for comparison, first frame (a): $P = 720$ and last frame (b): $P = 1440$.} \label{fig:plantFBP}
\end{figure}
\begin{table}[!bth]
\caption{Error metrics of the living tree reconstructions. All metrics are the mean of the first and last frames compared against the corresponding dense angle FBP reconstruction. Reconstructions are shown in figures~~\ref{fig:plant90angles} and~\ref{fig:plant30angles}.} \label{tab:plant}
\begin{tabular}{@{}llllll}
\specialrule{.1em}{.05em}{.05em} 
Method & $P$ & $\ell^2$-error & PSNR & HPSI \\
\hline
Haar CWDS-PDFP & 30 & 45.7\% & 15.4 & 0.407 \\
 & 45 & 46.2\% & 15.3 & 0.407 \\
 & 90 & 44.2\% & 15.7 & 0.439 \\ \hline
SH$_\text{2D}$ CWDS-PDFP & 30 & \HL{43.7\%} & \HL{15.8} & \HL{0.415} \\
 & 45 & \HL{43.9\%} & \HL{15.7} & \HL{0.436} \\
 & 90 & 43.7\% & 15.8 & 0.452 \\ \hline
SH$_\text{3D}$ CWDS-PDFP & 30 & 45.2\% & 15.5 & 0.393 \\
 & 45 & 45.2\% & 15.5 & 0.411 \\
 & 90 & \HL{43.7\%} & \HL{15.7} & \HL{0.439} \\
\specialrule{.1em}{.05em}{.05em} 
\end{tabular}
\end{table}

\begin{figure}[tbh!]
\begin{tabular}{r@{\hskip 3pt}ccc@{\hskip 3pt}l}
& $\tau = 1 \sim 0$ min & $\tau = 7 \sim 166$ min & $\tau = 11 \sim 235$ min & \\
\rotatebox[origin=l]{90}{\hspace{0.08\columnwidth} Haar wavelet}
& \includegraphics[width=0.3\columnwidth]{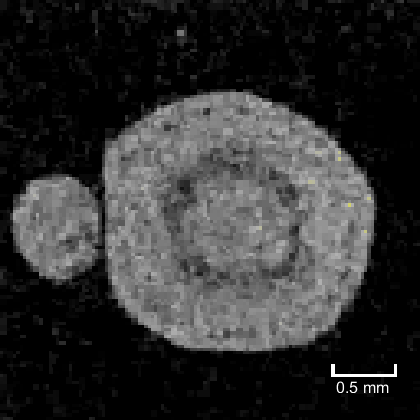}
& \includegraphics[width=0.3\columnwidth]{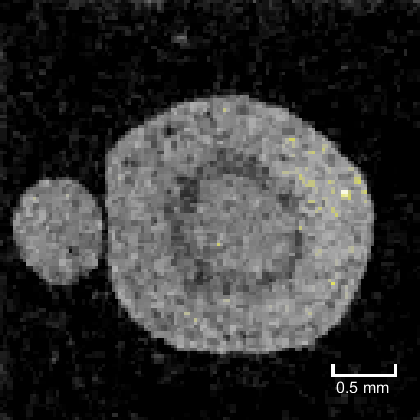}
& \includegraphics[width=0.3\columnwidth]{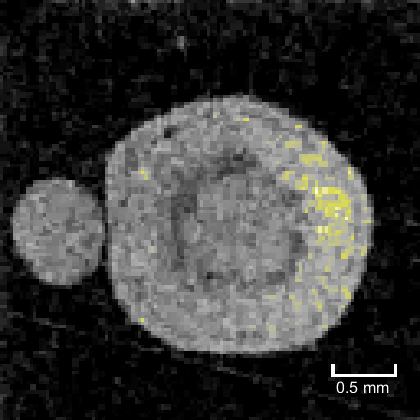}
& \includegraphics[height=0.30\columnwidth]{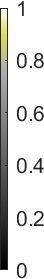} \\[6pt]
\rotatebox[origin=l]{90}{\hspace{0.09\columnwidth} 2D shearlet}
& \includegraphics[width=0.3\columnwidth]{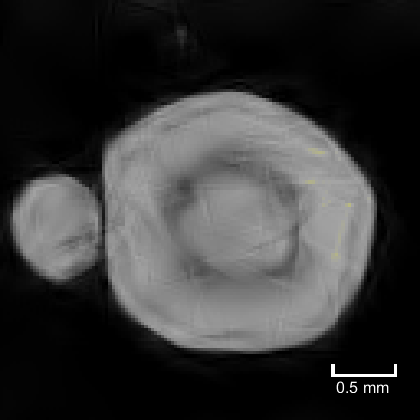}
& \includegraphics[width=0.3\columnwidth]{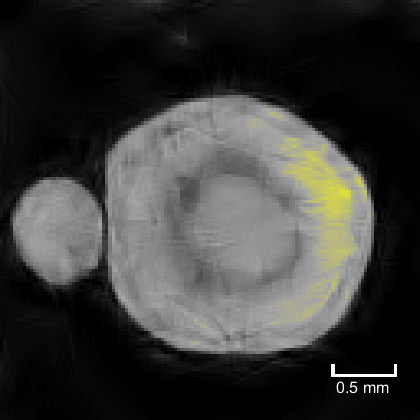}
& \includegraphics[width=0.3\columnwidth]{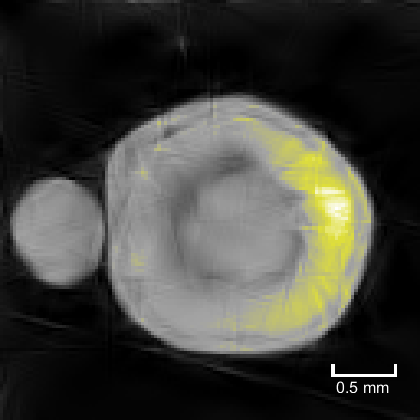}
& \includegraphics[height=0.30\columnwidth]{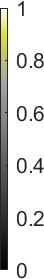} \\[6pt]
\rotatebox[origin=l]{90}{\hspace{0.09\columnwidth} 3D shearlet}
& \includegraphics[width=0.3\columnwidth]{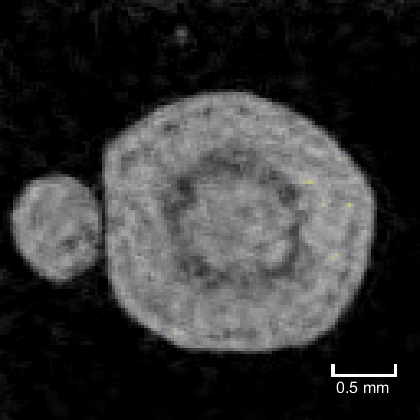}
& \includegraphics[width=0.3\columnwidth]{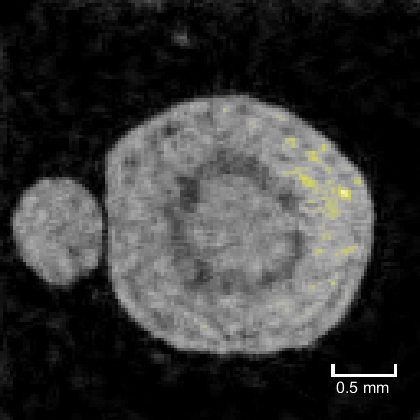}
& \includegraphics[width=0.3\columnwidth]{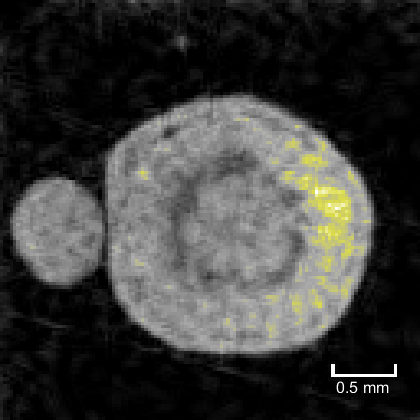}
& \includegraphics[height=0.30\columnwidth]{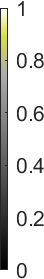}
\end{tabular}
\caption{Living tree reconstructed with different methods and $P = 90$.  On each row, time frames $\tau=1$ (leftmost), $\tau=7$ (middle) and $\tau=11$ (rightmost).} 
\label{fig:plant90angles}
\end{figure}
\begin{figure}[tbh!]
\begin{tabular}{r@{\hskip 3pt}ccc@{\hskip 3pt}l}
& $\tau = 1 \sim 0$ min & $\tau = 7 \sim 166$ min & $\tau = 11 \sim 235$ min & \\
\rotatebox[origin=l]{90}{\hspace{0.08\columnwidth} Haar wavelet}
& \includegraphics[width=0.3\columnwidth]{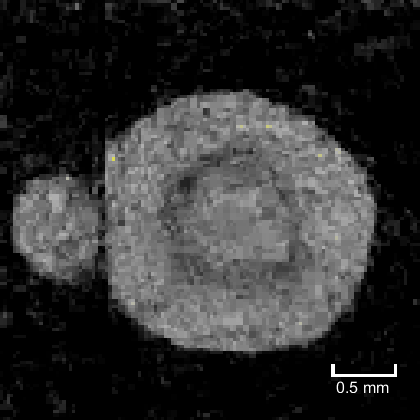}
& \includegraphics[width=0.3\columnwidth]{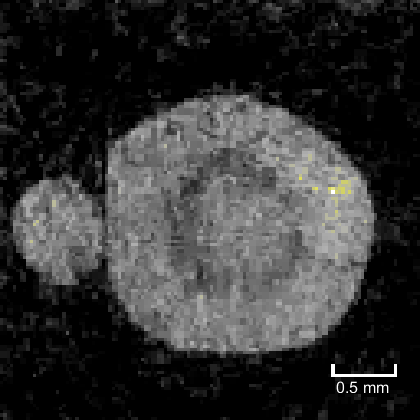}
& \includegraphics[width=0.3\columnwidth]{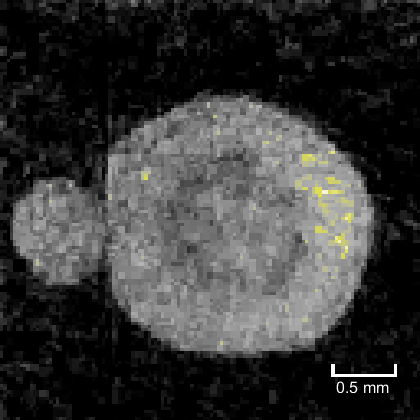}
& \includegraphics[height=0.30\columnwidth]{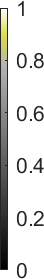} \\[6pt]
\rotatebox[origin=l]{90}{\hspace{0.09\columnwidth} 2D shearlet}
& \includegraphics[width=0.3\columnwidth]{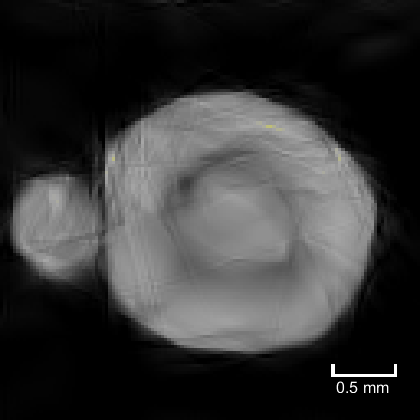}
& \includegraphics[width=0.3\columnwidth]{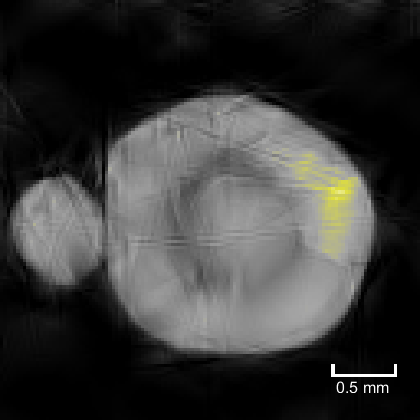}
& \includegraphics[width=0.3\columnwidth]{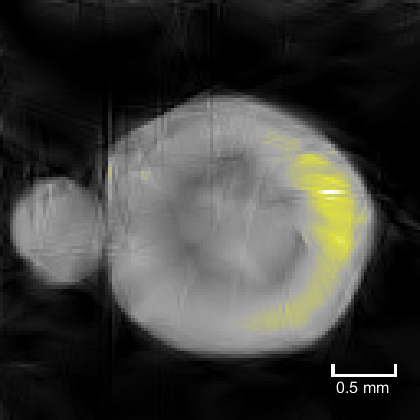}
& \includegraphics[height=0.30\columnwidth]{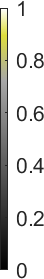} \\[6pt]
\rotatebox[origin=l]{90}{\hspace{0.09\columnwidth} 3D shearlet}
& \includegraphics[width=0.3\columnwidth]{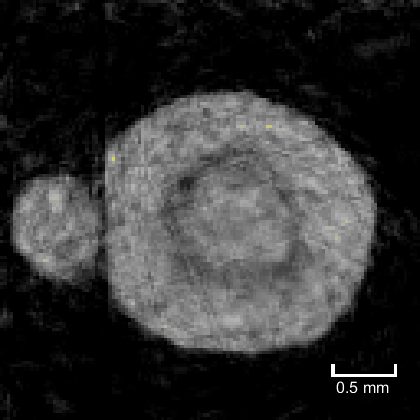}
& \includegraphics[width=0.3\columnwidth]{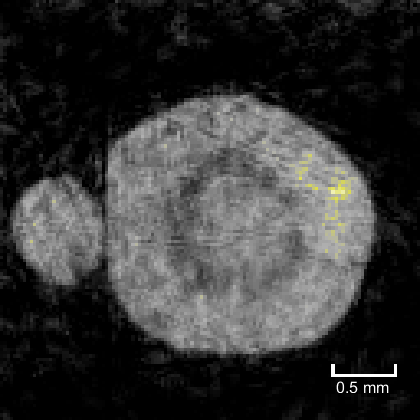}
& \includegraphics[width=0.3\columnwidth]{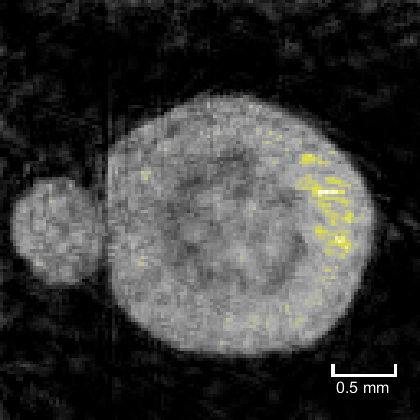}
& \includegraphics[height=0.30\columnwidth]{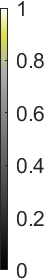}
\end{tabular}
\caption{Living tree reconstructed with different methods and $P = 30$.  On each row, time frames $\tau=1$ (leftmost), $\tau=7$ (middle) and $\tau=11$ (rightmost).} \label{fig:plant30angles}
\end{figure}
\begin{table}[!htb]
\caption{Iterations and computing time required for each reconstruction of the living tree. Note that sufficient number of projections for FBP was only available for the first and last frame. Reconstructions are shown in figures~~\ref{fig:plant90angles} and~\ref{fig:plant30angles}.} \label{tab:plantTimetable}
\begin{tabular}{@{}llllll}
\specialrule{.1em}{.05em}{.05em} 
Method & $T$ & $P$ & \begin{tabular}{@{}l@{}}Total \\ time (s)\end{tabular} & \begin{tabular}{@{}l@{}}Number of \\ iterations\end{tabular} & \begin{tabular}{@{}l@{}}Time per \\ iteration (s)\end{tabular} \\
\hline
 FBP & 1$^*$ & 720 & 0.2 & \ -- & \ -- \\
 & & 1440 & 0.2 & \ -- & \ -- \\ \hline
Haar CWDS-PDFP & 11 & 30 & 58 & 73  & 0.8 \\
 & & 45 & 90 & 85  & 1.1 \\
 & & 90 & 182 & 100  & 1.8 \\ \hline
SH$_\text{2D}$ CWDS-PDFP & 11 & 30 & 207 & 101 & 2.1 \\
 & & 45 & 259 & 107 & 2.4 \\
 & & 90 & 359 & 114  & 3.1 \\ \hline
SH$_\text{3D}$ CWDS-PDFP & 33 & 30 & 1287 & 62  & 20.7 \\
 & & 45 & 1464 & 68  & 21.5 \\
 & & 90 & 1489 & 75  & 19.8 \\
\specialrule{.1em}{.05em}{.05em}
\end{tabular}

{\small $^*$ Separate reconstructions of the first and last time frame only.}
\end{table}

\section{Conclusion}
\label{sec:concl}
In the present paper, we introduced a reconstruction method for sparse dynamic tomography.  The aim was to illustrate that taking into account the ongoing motion, rather than reconstructing each time frame separately, allows for better quality reconstructions, especially when only few projection angles are considered.  The forward problem is modeled by a time-dependent Radon transform that is well-defined in a finite time setting and the motion is tracked down by using 3D shearlets, where shearlets in the third dimension are deployed to connect the different time instances. We have shown that such a motion-aware model admits minimizers, which can be attained with modern primal-dual techniques combined with a sparsity driven rule for an automatic choice of the regularization parameter. 

The proposed model has been applied to simulated and real phantoms with undersampled data. The need for scarcely sampled data comes from the imaging of phloem transport in plants: phloem transport is inherently a dynamic process and full microCT scans are too slow to caption the movement of iodine in tissues and may harm the the plant, leading to distortion and artifacts in the reconstruction. However, distinction of finer structures needs higher resolution for accurate reconstructions, which justifies the use of slower but higher resolution equipment.

Overall, our experiments indicate that while there are clear differences between the three tested data sets, the results seem consistent: the dynamic 3D shearlets approach and the static 2D wavelets strategy provide good reconstructions in the simulated and the gel phantom cases, while the static 2D shearlets blur out most of the details for all three data sets. For the living tree case, the proposed dynamic 3D shearlets approach is the only one picking up the right information on the iodine spreading. Especially with the lowest number of projections ($P=30, 45$), the static approach starts to break down (wavelets cannot suppress noise, 2D shearlets are too blurry, smearing the iodine information with noise), while the proposed approach is robust also in the limit case.

When the computational burden is considered, it is clear that the 3D shearlets are noticeably slower to compute making our approach the more computationally demanding. The 2D shearlets are only slightly slower than wavelets, however in practise the wavelet based algorithm provided somewhat more robust process with different starting parameters while the shearlets often required a bit more tuning. 

More importantly, the results of these methods seem very promising to develop microCT as a tool to study phloem transport using iodine tracer. In particular the 3D shearlet and Haar wavelet seem to offer a reasonable signal to noise ratio. The 90 angles scans in the gel phantom and living tree experiments capture well the time dynamic of the tracer appearance in the samples and its spatial distribution,  which are critical to study phloem transport. Furthermore, with a measurement time of three minutes, we have the potential to capture events with a time resolution compatible with plant physiological processes.

For future studies the following improvements to the data sets should be considered: for simulated data each individual sinogram was generated from a static image and hence each projection is acquired at the same instant. This is of course false with the real measurements where the target can change during the measurement process. This could be improved by simulating only a portion of the projections and then slightly rotating (or changing features in) the phantom.

The real data sets might not accurately represent the effect of lower number of projections. Since these are obtained by downsampling a denser set of measurements, the total measurement time stays constant which is not desirable. However, obtaining separate but comparable measurements for each number of projections is difficult to replicate consistently, especially with real plants.

Finally, 4D reconstructions (3 spatial dimensions + time) would be of great importance to the end-user. While continuous shearlets can be extended to any dimension~\cite{Dahlke2010}, 4D discrete shearlets are not currently provided by any software package. Possible solutions would be to consider only the 3 spatial dimensions and to include, for example, a 1D wavelet transform in the time domain to obtain a 4D representation.

\section*{Acknowledgments}
TAB and SS acknowledge support by the Academy of Finland through the Finnish Centre of Excellence in Inverse Modelling and Imaging 2018-2025, decision number 312339, and Project 310822. SH, HH and TH acknowledge support by the Academy of Finland, Project 295696. YS acknowledges support by the Academy of Finland, decision number 323843 and 312571 and University of Helsinki, Faculty of Science ATMATH project.

The authors wish to thank Teemu H\"{o}ltt\"{a} and Luca Ratti for helpful discussions and Heikki Suhonen for recovering information about the data sets.

\bibliographystyle{amsplain}
\bibliography{main}

\providecommand{\bysame}{\leavevmode\hbox to3em{\hrulefill}\thinspace}
\providecommand{\MR}{\relax\ifhmode\unskip\space\fi MR }
\providecommand{\MRhref}[2]{%
  \href{http://www.ams.org/mathscinet-getitem?mr=#1}{#2}
}
\providecommand{\href}[2]{#2}
\begin{thebibliography}{10}

\bibitem{Bonnet03}
S.~Bonnet, A.~Koenig, S.~Roux, P.~Hugonnard, R.~Guillemaud, and P.~Grangeat,
  \emph{Dynamic x-ray computed tomography}, Proceedings {IEEE} \textbf{91}
  (2003), no.~10, 1574--1587.

\bibitem{Bubba19}
T.A. Bubba, G.~Kutyniok, M.~Lassas, M.~M{\"{a}}rz, W.~Samek, S.~Siltanen., and
  V.~Srinivasan, \emph{Learning the {I}nvisible: A hybrid deep
  learning-shearlet framework for limited angle computed tomography}, Inverse
  Problems \textbf{35} (2019), 064002.

\bibitem{Bubba17}
T.A. Bubba, M.~M{\"{a}}rz, Z.~Purisha, M.~Lassas, and S.~Siltanen,
  \emph{Shearlet-based regularization in sparse dynamic tomography},
  Proceedings of SPIE \textbf{10394} (2017), no.~103940Y-1, 1--10.

\bibitem{burger2017}
M.~Burger, H.~Dirks, L.~Frerking, A.~Hauptmann, T.~Helin, and S.~Siltanen,
  \emph{A variational reconstruction method for undersampled dynamic x-ray
  tomography based on physical motion models}, Inverse Problems \textbf{33}
  (2017), no.~12, 124008.

\bibitem{Candes2002}
E.~J. Cand\`{e}s and D.~L. Donoho, \emph{{New Tight Frames of Curvelets and
  Optimal Representations of Objects with Piecewise $C^2$ Singularities}},
  Commun. Pur. Appl. Math. (2002), 219--266.

\bibitem{cayla2015}
T.~Cayla, B.~Batailler, R.~Le Hir, F.~Revers, J.~Anstead, G.~Thompson,
  O.~Grandjean, and S.~Dinant, \emph{Live imaging of companion cells and sieve
  elements in {Arabidopsis} leaves}, PLoS One \textbf{10} (2015), no.~2,
  e0118122.

\bibitem{chen2016}
P.~Chen, J.~Huang, and X.~Zhang, \emph{A primal-dual fixed point algorithm for
  minimization of the sum of three convex separable functions}, Fixed Point
  Theory and Applications \textbf{2016} (2016), no.~1, 54.

\bibitem{Cormack63}
A.M. Cormack, \emph{Representation of a function by its line integrals, with
  some radiological applications i}, Journal of Applied Physics \textbf{34}
  (1963), 2722--2727.

\bibitem{Dahlke2010}
S.~Dahlke, G.~Steidl, and G.~Teschke, \emph{The continuous shearlet transform
  in arbitrary space dimensions}, Journal of Fourier Analysis and Applications
  \textbf{16} (2010), no.~3, 340--364.

\bibitem{daubechies2004}
I.~Daubechies, M.~Defrise, and C.~De Mol, \emph{An iterative thresholding
  algorithm for linear inverse problems with a sparsity constraint},
  Communication on Pure and Applied Mathematics \textbf{57} (2004), 1413--1457.

\bibitem{Donoho2001}
D.~L. Donoho, \emph{{Sparse Components of Images and Optimal Atomic
  Decompositions}}, Constr. Approx. \textbf{17} (2001), no.~3, 353--382.

\bibitem{earles2018}
J.M. Earles, T.~Knipfer, A.~Tixier, J.~Orozco, C.~Reyes, M.~Zwieniecki,
  C.~Brodersen, and A.~McElrone, \emph{In vivo quantification of plant starch
  reserves at micrometer resolution using x-ray micro {CT} imaging and machine
  learning}, New Phytologist \textbf{218} (2018), no.~3, 1260--1269.

\bibitem{epron2015}
D.~Epron, O.~Cabral, J-P. Laclau, M.~Dannoura, A.~Packer, C.~Plain,
  P.~Battie-Laclau, M.~Moreira, P.~Trivelin, J-P. Bouillet, et~al., \emph{In
  situ 13co2 pulse labelling of field-grown eucalypt trees revealed the effects
  of potassium nutrition and throughfall exclusion on phloem transport of
  photosynthetic carbon}, Tree physiology \textbf{36} (2015), no.~1, 6--21.

\bibitem{Grohs2011}
P.~Grohs, \emph{Continuous shearlet frames and resolution of the wavefront
  set}, Monats. Math. \textbf{164} (2011), no.~4, 393--426.

\bibitem{Hahn14}
B.~Hahn, \emph{Reconstruction of dynamic objects with affine deformations in
  computerized tomography}, Journal of Inverse and Ill-posed Problems
  \textbf{22} (2014), 323--339.

\bibitem{Hahn16}
\bysame, \emph{Null space and resolution in dynamic computerized tomography},
  Inverse Problems \textbf{32} (2016), 025006.

\bibitem{hubeau2015}
M.~Hubeau and K.~Steppe, \emph{Plant-{PET} scans: in vivo mapping of xylem and
  phloem functioning}, Trends in plant science \textbf{20} (2015), no.~10,
  676--685.

\bibitem{Engl1996}
M.~Hanke H.W.~Engl and A.~Neubauer, \emph{Regularization of inverse problems},
  Springer Science \& Business Media, 1996.

\bibitem{Hakkarainen19}
J.~Hakkarainen J, Z.~Purisha, A.~Solonen, and S.~Siltanen, \emph{Undersampled
  dynamic x-ray tomography with dimension reduction kalman filter}, {IEEE}
  Transactions on Computational Imaging \textbf{5} (2019), no.~3, 492--5012.

\bibitem{Katsevich10}
A.~Katsevich, \emph{An accurate approximate algorithm for motion compensation
  in two-dimensional tomography}, Inverse Problems \textbf{26} (2010), no.~6,
  065007.

\bibitem{keyes2017}
S.~Keyes, N.~Gostling, J.~Cheung, T.~Roose, I.~Sinclair, and A.~Marchant,
  \emph{The application of contrast media for in vivo feature enhancement in
  x-ray computed tomography of soil-grown plant roots}, Microscopy and
  Microanalysis \textbf{23} (2017), no.~3, 538--552.

\bibitem{Kittipoom2012}
P.~Kittipoom, G.~Kutyniok, and W.-Q. Lim, \emph{Construction of compactly
  supported shearlet frames}, Constructive Approximation \textbf{35} (2012),
  no.~1, 21--72.

\bibitem{knipfer2017}
T.~Knipfer, I.~Cuneo, J.M. Earles, C.~Reyes, C.~Brodersen, and A.~McElrone,
  \emph{Storage compartments for capillary water rarely refill in an intact
  woody plant}, Plant physiology \textbf{175} (2017), no.~4, 1649--1660.

\bibitem{kutyniok2009}
G.~Kutyniok and D.~Labate, \emph{{Resolution of the Wavefront Set Using
  Continuous Shearlets}}, T. Am. Math. Soc. \textbf{361} (2009), no.~5,
  2719--2754.

\bibitem{Kutyniok2012}
G.~Kutyniok and D.~Labate (eds.), \emph{Shearlets, multiscale analysis for
  multivariate data}, Birkh\"{a}user Basel, 2012.

\bibitem{Kutyniok2012a}
G.~Kutyniok, J.~Lemvig, and W.-Q. Lim, \emph{Optimally sparse approximations of
  {3D} functions by compactly supported shearlet frames}, SIAM J. Math. Anal.
  \textbf{44} (2012), 2962--3017.

\bibitem{kutyniokch5}
\bysame, \emph{Shearlets and optimally sparse approximations}, Shearlets,
  Multiscale Analysis for Multivariate Data, Birkhäuser Basel, 2012,
  pp.~145--198.

\bibitem{Kutyniok2016}
G.~Kutyniok, W.-Q Lim, and R.~Reisenhofer, \emph{Shear{L}ab {3D}: Faithful
  digital shearlet transforms based on compactly supported shearlets}, ACM
  Trans. Math. Software \textbf{42} (2016), no.~5, 42.

\bibitem{lemoine2013}
R.~Lemoine, S.~La Camera, R.~Atanassova, F.~D{\'e}dald{\'e}champ, T.~Allario,
  N.~Pourtau, J-L. Bonnemain, M.~Laloi, P.~Coutos-Th{\'e}venot, L.~Maurousset,
  et~al., \emph{Source-to-sink transport of sugar and regulation by
  environmental factors}, Frontiers in plant science \textbf{4} (2013), 272.

\bibitem{Li05}
T.~Li, E.~Schreibmann, Y.~Yang, and L.~Xing, \emph{Motion correction for
  improved target localization with on-board cone-beam computed tomography},
  Physics in medicine and biology \textbf{51} (2005), 253.

\bibitem{Mallat2009}
S.~Mallat, \emph{{A Wavelet Tour of Signal Processing: The Sparse Way}}, 3rd
  ed., Elsevier, Amsterdam, 2009.

\bibitem{wavelettoolbox2018}
MathWorks, \emph{{Wavelet Toolbox for Matlab, v5.0}}, 2018, (Accessed:
  1.7.2019).

\bibitem{mencuccini2013}
M.~Mencuccini, T.~H{\"o}ltt{\"a}, S.~Sevanto, and E.~Nikinmaa, \emph{Concurrent
  measurements of change in the bark and xylem diameters of trees reveal a
  phloem-generated turgor signal}, New phytologist \textbf{198} (2013), no.~4,
  1143--1154.

\bibitem{natterer2001}
F.~Natterer, \emph{The mathematics of computerized tomography}, {SIAM:}
  {S}ociety for {I}ndustrial and {A}pplied {M}athematics, Philadelphia, PA
  (USA), 2001.

\bibitem{nattererwubbeling01}
F.~Natterer and F.~W{\"{u}}bbeling, \emph{Mathematical methods in image
  reconstruction}, {SIAM}: {S}ociety for {I}ndustrial and {A}pplied
  {M}athematics, Philadelphia, PA (USA), 2001.

\bibitem{Niemi15}
E.~Niemi, M.~Lassas, A.~Kallonen, L.~Harhanen, K.~H{\"{a}}m{\"{a}}l{\"{a}}inen,
  and S.~Siltanen, \emph{Dynamic multi-source {X}-ray tomography using a
  spacetime level set method}, Journal of Computational Physics \textbf{291}
  (2015), 218--237.

\bibitem{purisha2017}
Z.~Purisha, J.~Rimpel{\"a}inen, T.~Bubba, and S.~Siltanen, \emph{Controlled
  wavelet domain sparsity for x-ray tomography}, Measurement Science and
  Technology \textbf{29} (2017), no.~1, 014002.

\bibitem{reisenhofer2018}
R.~Reisenhofer, S.~Bosse, G.~Kutyniok, and T.~Wiegand, \emph{A {H}aar
  wavelet-based perceptual similarity index for image quality assessment},
  Signal Processing: Image Communication \textbf{61} (2018), 33--43.

\bibitem{Ritman03}
Erik~L. Ritman, \emph{Cardiac computed tomography imaging: a history and some
  future possibilities}, Cardiol. Clin. \textbf{21} (2003), no.~4, 491--513.

\bibitem{Roux04}
S.~Roux, L.~Desbat, A.~Koenig, and P.~Grangea, \emph{Exact reconstruction in
  {2D} dynamic {CT}: Compensation of time-dependent affine deformations}, Phys.
  Med. Biol. \textbf{49} (2004), 2169--2182.

\bibitem{royden2010}
H.~Royden and P.~Fitzpatrick, \emph{Real analysis, 4th edition}, Prentice Hall,
  2010.

\bibitem{salmon2019}
Y.~Salmon, L.~Dietrich, S.~Sevanto, T.~H{\"o}ltt{\"a}, M.~Dannoura, and
  D.~Epron, \emph{Drought impacts on tree phloem: from cell-level responses to
  ecological significance}, Tree physiology \textbf{39} (2019), no.~2,
  173--191.

\bibitem{savage2016}
J.~Savage, M.~Clearwater, D.~Haines, T.~Klein, M.~Mencuccini, S.~Sevanto,
  R.~Turgeon, and C.~Zhang, \emph{Allocation, stress tolerance and carbon
  transport in plants: how does phloem physiology affect plant ecology?},
  Plant, Cell \& Environment \textbf{39} (2016), no.~4, 709--725.

\bibitem{savage2013}
J.~Savage, M.~Zwieniecki, and M.~Holbrook, \emph{Phloem transport velocity
  varies over time and among vascular bundles during early cucumber seedling
  development}, Plant Physiology \textbf{163} (2013), no.~3, 1409--1418.

\bibitem{Siltanen03}
Samuli Siltanen, Ville Kolehmainen, S.~Jarvenpaa, J.~Kaipio, P.~Koistinen,
  Matti Lassas, J.~Pirttila, and Erkki Somersalo, \emph{Statistical inversion
  for {X-}ray tomography with few radiographs {I}: General theory}, Physics in
  Medicine and Biology \textbf{48} (2003), 1437--1463.

\bibitem{Siltanen03part2}
\bysame, \emph{Statistical inversion for {X-}ray tomography with few
  radiographs {II}: Applications to dental radiology}, Physics in Medicine and
  Biology \textbf{48} (2003), 1465--1490.

\bibitem{suuronen2014}
J.-P. Suuronen, A.~Kallonen, V.~H{\"a}nninen, M.~Blomberg,
  K.~H{\"a}m{\"a}l{\"a}inen, and R.~Serimaa, \emph{Bench-top x-ray
  microtomography complemented with spatially localized x-ray scattering
  experiments}, Journal of Applied Crystallography \textbf{47} (2014), no.~1,
  471--475.

\bibitem{van2016}
W.~van Aarle, W.~Palenstijn, J.~Cant, E.~Janssens, F.~Bleichrodt,
  A.~Dabravolski, J.~De Beenhouwer, K.-J. Batenburg, and J.~Sijbers, \emph{Fast
  and flexible x-ray tomography using the {ASTRA} toolbox}, Optics express
  \textbf{24} (2016), no.~22, 25129--25147.

\bibitem{spot2012}
E.~van~den Berg and M.~Friedlander, \emph{{Spot -- A Linear-Operator Toolbox,
  v1.2}}, 2012, (Accessed: 20.6.2019).

\bibitem{windt2006}
C.~Windt, F.~Vergeldt, A.~De Jager, and H.~Van As, \emph{{MRI} of long-distance
  water transport: a comparison of the phloem and xylem flow characteristics
  and dynamics in poplar, castor bean, tomato and tobacco}, Plant, Cell \&
  Environment \textbf{29} (2006), no.~9, 1715--1729.

\end{thebibliography}

\end{document}